\pdfoutput=1
\documentclass[runningheads]{llncs}
\usepackage[T1]{fontenc}
\usepackage{graphicx}
\usepackage{xcolor}
\usepackage{amsmath}
\usepackage{amssymb}
\usepackage{enumitem}
\usepackage[frozencache,cachedir=.]{minted}
\usepackage{algorithm}
\usepackage[noend]{algpseudocode}
\usepackage{xspace}
\usepackage{tabularx}
\usepackage{booktabs} 
\usepackage{multirow}
\usepackage{listings}
\usepackage{caption}
\usepackage{subcaption}

\usepackage{hyperref}
\algrenewcommand\algorithmicindent{1.0em}

\newcommand{\sayan}[1]{\textcolor{blue}{#1}}
\newcommand{\yangge}[1]{\textcolor{black}{#1}}

\newcommand{\scenario}{{SC}}
\newcommand{\agent}{\mathcal{A}}
\newcommand{\map}{\mathcal{M}}
\newcommand{\sensor}{\mathcal{S}}
\newcommand{\mapvert}{{V}}
\newcommand{\maplabel}{{L}\xspace}
\newcommand{\maplabcurve}{{g}\xspace}
\newcommand{\maplablab}{{h}\xspace}
\newcommand{\mapcurve}{{\omega}\xspace}
\newcommand{\mapsetcurve}{{\Omega}\xspace}

\newcommand{\strategy}{{P}\xspace}

\newcommand{\Zagent}{{Z}}

\newcommand{\Xagent}{{X}}
\newcommand{\xagent}{{x}}
\newcommand{\Dagent}{{D}}
\newcommand{\dagent}{{d}}

\newcommand{\agentflow}{{F}}

\newcommand{\agentfullstate}{{Y}}
\newcommand{\workspace}{{W}}

\newcommand{\ourtool}{{Verse}\xspace}

\newcommand{\haguard}{\textbf{G}}
\newcommand{\hareset}{\textbf{R}}
\newcommand{\ha}{{H}}
\newcommand{\reals}{{\mathbb{R}}}

\newcommand{\Dha}{\textbf{D}\xspace}
\newcommand{\dha}{\textbf{d}\xspace}
\newcommand{\Xha}{\textbf{X}}
\newcommand{\Xhaset}{\textbf{S}}
\newcommand{\xha}{\textbf{x}}
\newcommand{\Tha}{\xi}
\newcommand{\transvert}{\textbf{V}}

\newcommand{\tl}{\textbf{TL}}

\newcommand{\exec}{{\alpha}}
\newcommand{\reach}[1]{{Reach_{\ha_{#1}}}\xspace}
\newcommand{\unsafe}{\textbf{U}}

\newcommand{\agentguard}{{G}\xspace}
\newcommand{\agentreset}{{R}\xspace}

\newcommand{\post}{{post}\xspace}
\newcommand{\postcont}{\texttt{postCont}\xspace}
\newcommand{\postdisc}{\texttt{postDisc}\xspace}

\newcommand{\exectree}{{Tree}\xspace}
\newcommand{\treevert}{{V}\xspace}
\newcommand{\treeedge}{{E}\xspace}

\newcommand{\guardcache}{{C_g}\xspace}
\newcommand{\flowcache}{{C_f}\xspace}
\newcommand{\unknown}{{\texttt{unknown}}\xspace}
\newcommand{\sat}{{\texttt{sat}}\xspace}
\newcommand{\unsat}{{\texttt{unsat}}\xspace}

\newcommand{\none}{\texttt{none}\xspace}
\newcommand{\treedepth}{{m}\xspace}
\newcommand{\algveri}{{\texttt{verify}}\xspace}
\newcommand{\alginc}{{\texttt{verifyInc}}\xspace}
\newcommand{\algfcache}{{\texttt{flowCache}}\xspace}

\newcommand{\sd}[1]{{\langle \Xhaset_{#1},\dha_{#1}\rangle}\xspace}
\newcommand{\sdp}[1]{{\langle \Xhaset_{#1}',\dha_{#1}'\rangle}\xspace}
\newcommand{\sdpp}[1]{{\langle \Xhaset_{#1}'',\dha_{#1}''\rangle}\xspace}

\newcommand{\verify}{\texttt{verify}\xspace}
\newcommand{\simulate}{\texttt{simulate}\xspace}
\newcommand{\assert}{\texttt{assert}\xspace}

\definecolor{bg}{rgb}{0.95,0.95,0.95}

\newcounter{code}

\usepackage{newfloat}

\DeclareFloatingEnvironment[
  fileext = lop ,
  listname = {Code Snippet} ,
  within=none,
  name = Code Snippet,
  placement = h
]{codesnip}

\begin{document}

\title{Verse: A Python library for reasoning about multi-agent hybrid system scenarios}
%
%
\author{Yangge Li\inst{1} \and
Haoqing Zhu\inst{1} \and
Katherine Braught\inst{1}\and 
Keyi Shen\inst{1}\and
Sayan Mitra\inst{1}}

\authorrunning{}

\institute{Coordinated Science Laboratory\\ University of Illinois at Urbana-Champaign \\ 
\email{\{li213, haoqing3, braught2, keyis2, mitras\}@illinois.edu}}
\maketitle              
\begin{abstract}
We present the  \ourtool library with the aim of making hybrid system verification  more usable for multi-agent scenarios. 
In \ourtool, decision making agents move in a map and interact with each other through sensors. The decision logic for each agent is written in a subset of Python and the continuous dynamics is given by a black-box simulator. 
Multiple agents can be instantiated and they can be ported to different maps for creating scenarios. 
\ourtool provides functions for  simulating and verifying such scenarios using existing reachability analysis algorithms. 
 We illustrate several capabilities and use cases of the library with heterogeneous agents, incremental verification, different sensor models, and the flexibility of plugging in different subroutines for post computations. 
\end{abstract}

\section{Introduction}
\label{sec:intro}
Hybrid system verification tools have been used to analyze linear  models with  thousands of continuous dimensions~\cite{bak2017hylaa,10.1145/3302504.3311792,Althoff2015a} and  nonlinear models inspired by industrial applications~\cite{10.1145/3302504.3311792,dryvr}. 
Chen and Sankaranarayanan provide a survey of the state of the art~\cite{10.1007/978-3-031-06773-0_6}.
Despite the large potential user base,  current usage of this  technology remains concentrated within the formal methods community. 
%
We conjecture that usability is one of the key barriers. 
Most hybrid verification tools~\cite{spaceEx,10.1007/978-3-319-26287-1_1,flow,bak2017hylaa,Althoff2015a}
require the input model to be written in a  tool-specific  language. 
Tools like C2E2~\cite{FanQM0D16} attempt to translate a subclass of models from  the popular Simulink/Stateflow  framework, 
but the  language-barrier goes deeper than syntax. The verification algorithms are based on variants of the hybrid automaton~\cite{alur95algorithmic,HENZINGER199894,TIOAmon} which require the discrete states (or {\em modes}) to be spelled out  explicitly as a graph, with guards and resets labeling the edges. 

In contrast,  the code for {\em simulating\/} a multi-agent scenario would be written in an expressive programming language. Each agent will have  a decision logic and some continuous dynamics. A complex scenario would be composed by putting together a collection of  agents; it  may use a  {\em map\/} which brings additional structure and constraints to the agent's decisions and interactions. Describing or translating such scenarios for hybrid verification is a  far cry from the capabilities of current tools. 

In this paper, we present \ourtool\footnote{We will make the tool available for artifact evaluations. We omitted the online link in this submission to avoid compromising the double blind review requirements.}, a Python library that aims to make hybrid technologies more usable for multi-agent scenarios. An agent's decision logic is  written in an expressive subset of Python (See Fig.~\ref{fig:controller}). This program can access the relevant features of the map  and parts of the states of the other agents. 
The  continuous dynamics of an agent has to be supplied as a black-box simulation function.
Multiple agents with different dynamics and different decision logics are instantiated to create a \ourtool scenario.

%

\ourtool provides functions for performing systematic simulation and verification of such scenarios through reachability analysis. An agent's decision logic can allow nondeterministic choices. If multiple Python \texttt{if} conditions are satisfied at a given state, then both branches are explored. 
%
For example, Fig.~\ref{fig:ndsim}({\em center}) shows simulations in which the red drone on track \texttt{T1} nears the blue and nondeterministically switches to tracks \texttt{T0} and \texttt{T2}. Similarly, the \verify function propagates uncertainty in the initial states through branches. Safety requirements written with {\assert} can be checked via reachability analysis. 
 While the functions for traversal and $\post$ computation are based on known algorithms, new ones can be implemented and the library vastly simplifies  specification of hybrid multi-agent scenarios. 
 %

The key  concept  that makes the above functionalities tractable and  specifications expressive is the {\em  map\/} abstraction. A \ourtool map  defines a set of {\em tracks\/} that agents can follow.  While a map may have infinitely many tracks, they fall in a finite number of {\em track modes\/}. For example, in  Fig.~\ref{fig:fig8map} each layer in the map is assigned to a track mode \yangge{(\texttt{T0-2})} and the tracks between each pair of layers are also assigned to a track mode (\texttt{M10}, \texttt{M01} etc.). Further, when an agent makes a decision and changes its internal mode (called {\em tactical mode}, which is different from the track mode), for example from \texttt{Normal} to \texttt{MoveUp}, the map object  determines the new track mode for the agent. For an agent on track \texttt{T1}, its new track mode will be \texttt{M10}. This map abstraction allows portability of an agent's decision logic across different maps as long as they have the same interface or track modes. 
It makes \ourtool suited for  multi-agent scenarios, arising in  motion control~\cite{netcontrol:Mobihoc04},  air-traffic management~\cite{utm_conops}, and the study of tactical collision avoidance~\cite{manfredi2016introduction}. 

In summary, the main contributions of this paper are:
 (1) an expressive framework for specifying and simulating  multi-agent hybrid scenarios on interesting maps (Section~\ref{sec:scenarios}),
(2) a library of powerful functions for building verification algorithms based on reachability analysis (Section~\ref{sec:verification_algorithm}), and
(3) illustrations of several capabilities and use cases of the library with heterogeneous agents, incremental verification, 
different sensor models, 
and  the flexibility of plugging in  different subroutines for {\em post\/} computations (Section~\ref{sec:exp}).


\paragraph{Related work.} 
There are a number of powerful tools for creating,  simulating, and testing complex multi-agent scenarios~\cite{scenic,Wu2017FlowAA,brittain2022aamgym,GRAICrace}.
For instance, Scenic \cite{scenic} uses a probabilistic programming language for  guided sampling, simulations, and falsification, Flow~\cite{Wu2017FlowAA} integrates the SUMO~\cite{SUMO2018} traffic simulator for  reinforcement learning, AAM-GYM~\cite{brittain2022aamgym} can generate and simulate scenarios for testing AI algorithms in advanced air mobility, and GRAIC~\cite{GRAICrace} has been used for testing racing controllers in dynamic environments. While the models created in these  tools can be  very flexible and expressive, they are not readily amenable to formal verification. 

Interactive theorem provers have been used for modeling and verification of multi-agent, distributed and hybrid systems~\cite{10.1007/978-3-030-90870-6_20,10.1007/978-3-319-21401-6_36,LKLM:formats05,OLVECZKY2002359}. 
Most notably \yangge{KeYmeraX \cite{10.1007/978-3-319-21401-6_36}} uses quantified differential dynamic logic for specifying multi-agent scenarios and supports speculative proof search and user defined tactics. Isabelle/HOL~\cite{10.1007/978-3-030-90870-6_20}, PVS~\cite{LKLM:formats05}, and Maude~\cite{OLVECZKY2002359} have also been used for limited classes of hybrid systems. These approaches require significant levels of user expertise to interact with. 

This work is closest to the tool presented in~\cite{sibai-tacas-2020}, which also supports multiple agents and reachability analysis. However, \ourtool significantly improves usability by allowing (a) Python for the  decision logics and (b) complex maps. 
The theoretical ideas of Sibai et al.~\cite{sibai2021scenechecker,sibai-tacas-2020}, in exploiting symmetries and caching are complementary to our contributions and could indeed be incorporated to improve the verification algorithms in \ourtool.



\section{Overview of \ourtool}
\label{sec:example}

 We will  highlight the key features of \ourtool with an example. 
Consider two drones flying along three parallel figure-eight tracks that are vertically separated in space (shown by black lines in  Fig.~\ref{fig:fig8map}). Each drone has a simple collision avoidance logic: if it gets too close to another drone on the same track, then it  switches to either the track above or the one below. A drone  on \texttt{T1} has both choices. 
%
%
\ourtool enables creation, simulation, and verification of such scenarios using \yangge{Python, and} provides a collection of powerful functions for building new analysis algorithms for such scenarios.  

\begin{figure}[t]
     \begin{minipage}{0.33\textwidth}
        \centering
        \includegraphics[width=\textwidth,height=2cm]{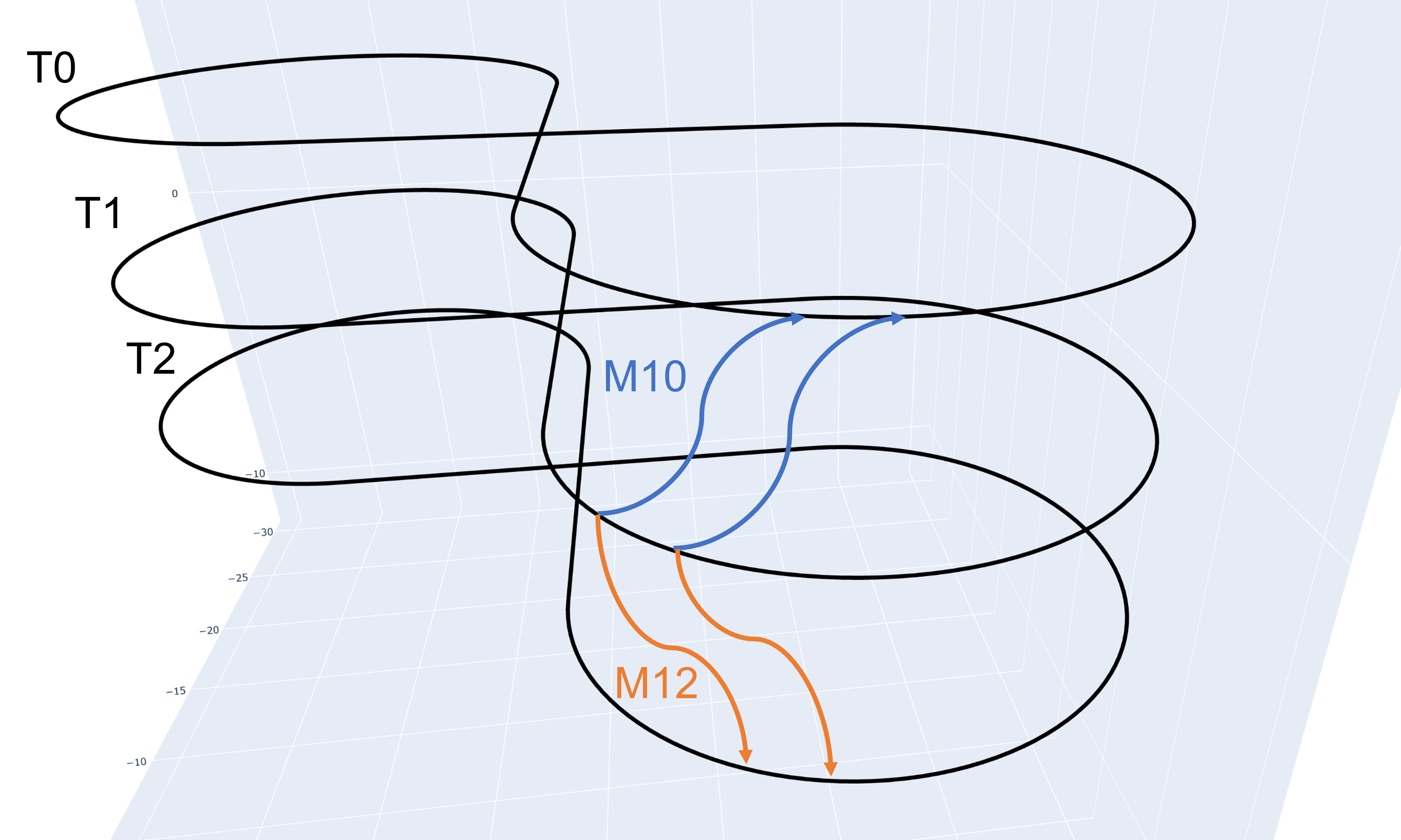}
    \end{minipage}\hfill
    \begin{minipage}{0.33\textwidth}
        \centering
        \includegraphics[width=\textwidth, trim=30 30 0 18,clip,height=2cm]{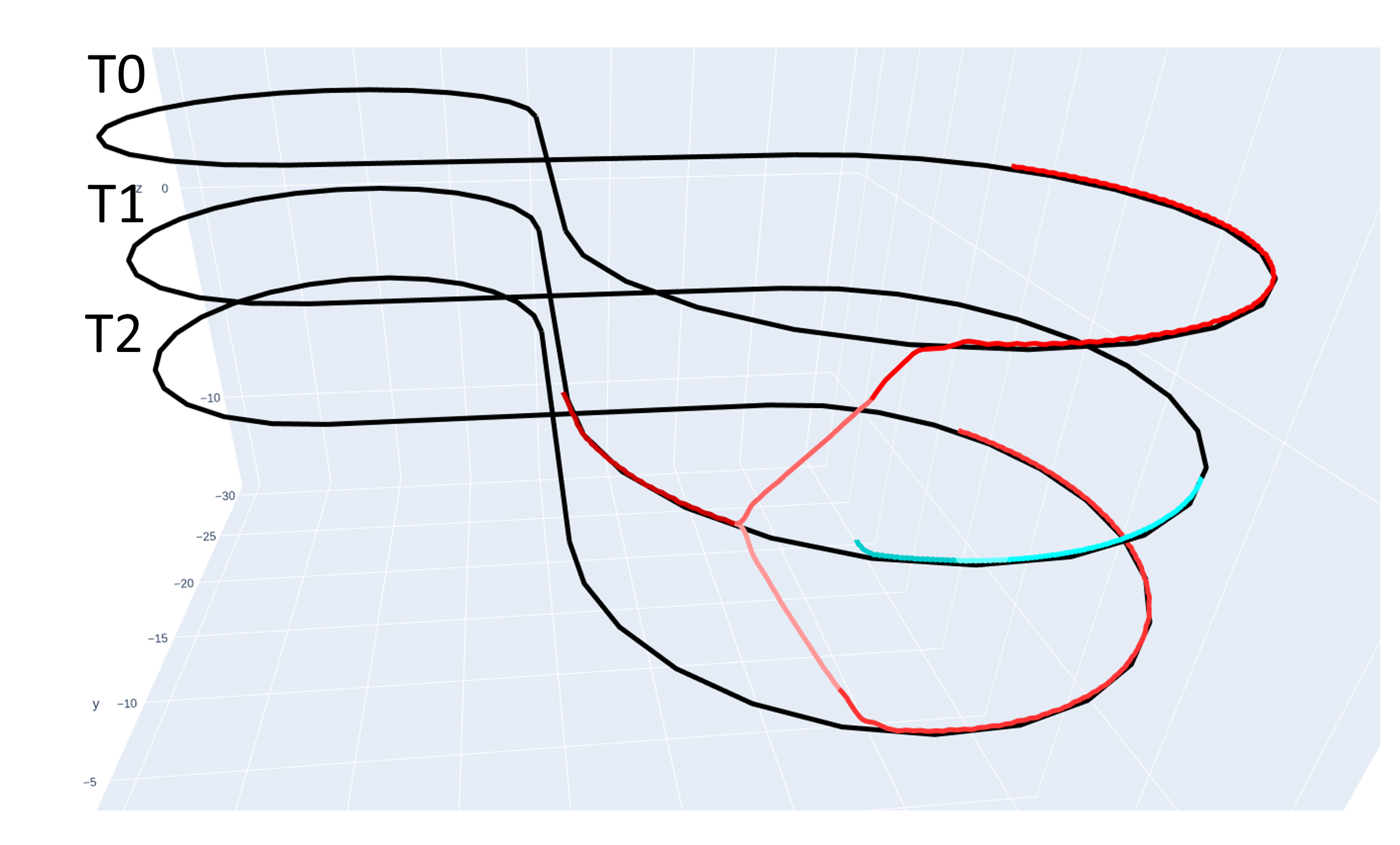}
    \end{minipage}\hfill
     \begin{minipage}{0.33\textwidth}
        \centering
            \includegraphics[width=\textwidth, trim=0 25 0 50,clip,height=2cm]{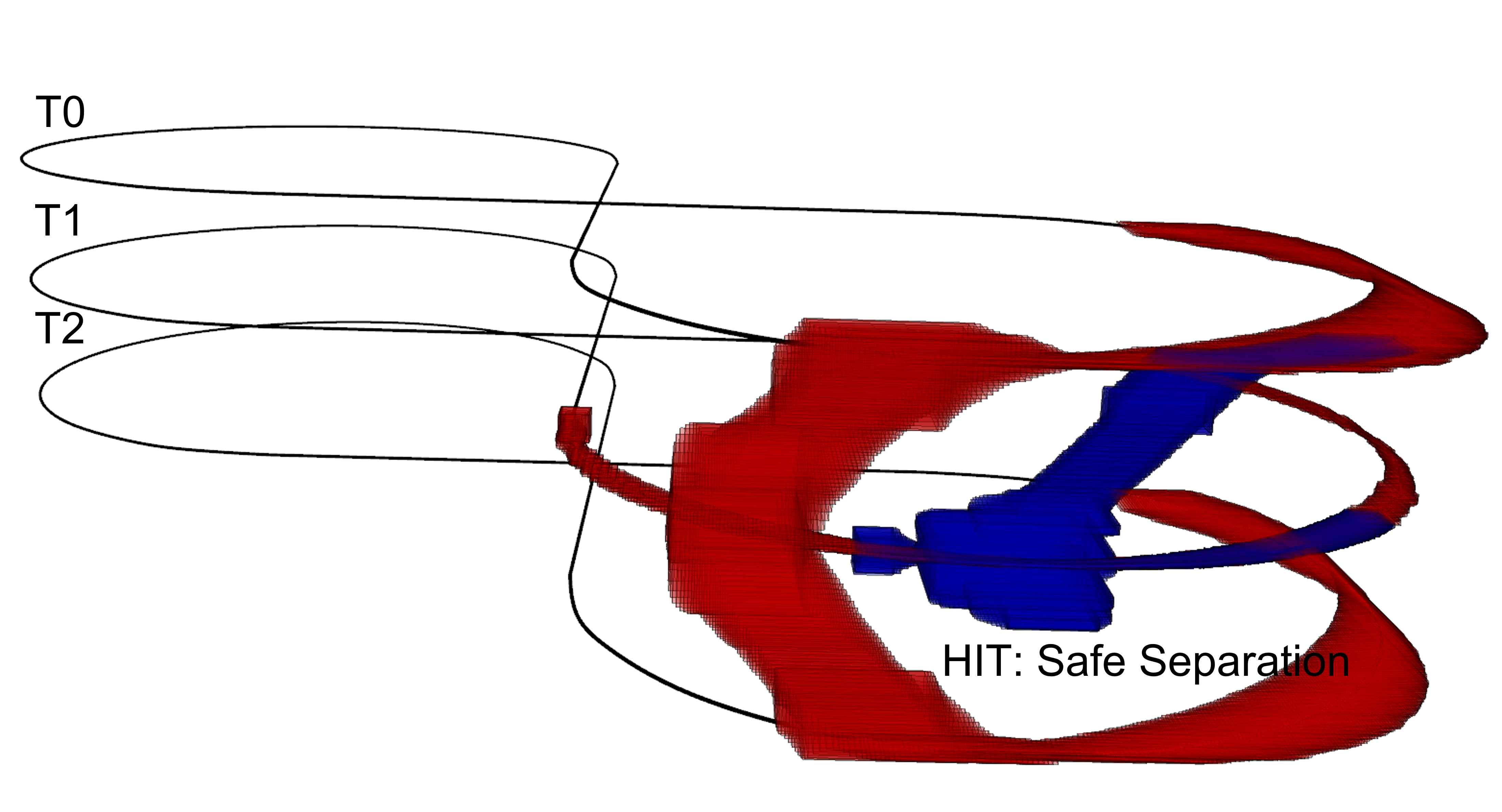}
    \end{minipage}\hfill
    \caption{\small {\em Left:} A 3-d figure-8 map with example track mode labels. {\em Center:} Simulation of a red drone nearing  the blue drone on \texttt{T1} and nondeterministically moving to \texttt{T0} or \texttt{T2}. Both branches are computed by \ourtool's \texttt{simulate} function. \textit{Right:} Computed reachable sets of the two drones cover more possibilities: both drones are allowed to switch tracks when they get close. All four  branches are explored by \ourtool and one is found to violate safety.
    }
    \label{fig:fig8map}
    \label{fig:demo_graph}
    \label{fig:ndsim}
\end{figure}

\paragraph{Creating scenarios.} Agents like the drones in this example are described by a {\em  simulator} and a {\em decision logic}. The  decision logic has to  be written in an expressive subset of Python (see code in Fig.~\ref{fig:controller} and  Appendix \ref{sec:parser_grammar} for more details). 
The decision logic for an ego agent takes as input its current state  and the (observable) states of the other agents, and updates the discrete state or the {\em mode}  of the ego agent. It may also update the continuous state of the ego agent.  
The {\em mode} of an agent, as we shall see later in Section~\ref{sec:agents}, has two parts---a {\em tactical mode\/} corresponding to agent's decision or discrete state, and a {\em track mode\/}  that is determined by the map. 
Using the \texttt{any} and \texttt{all} functions, the agent's decision logic can quantify over  other agents in the scene. For example, in lines~41-43 of Fig.~\ref{fig:controller} an agent updates its tactical mode to begin a track change if there is any agent near it.

\ourtool will parse this decision logic to internally construct the transition graph of the hybrid model with  guards and resets. 
The simulator can be written \yangge{in} any language and is treated as a black-box\footnote{\label{footnote:blackbox}This design decision for \ourtool is relatively independent. For reachability analysis, \ourtool currently uses black-box statistical approaches implemented in DryVR~\cite{dryvr} and NeuReach~\cite{neuReach-SunM22}. If the simulator is available as a white-box model, such as differential equations, then \ourtool could use model-based reachability analysis.}. For the examples discussed in this paper, the simulators are also written in Python.
Safety requirements can be specified using \assert statements
(see Fig. \ref{code:unsafe}). 
\begin{figure}
    \centering
    \scriptsize
\begin{minted}{python}
def decisionLogic(ego: State, others: List[State], track_map):
  next = copy.deepcopy(ego)
  if ego.tactical_mode == TacticalMode.Normal:
    if any((is_close(ego, other) and ego.track_mode==other.track_mode) for other in others):
      next.tactical_mode = TacticalMode.MoveDown
      next.track_mode = track_map.h(ego.track_mode, ego.tactical_mode, TacticalMode.MoveDown)
    if any((is_close(ego, other) and ego.track_mode==other.track_mode) for other in others):
      next.tactical_mode = TacticalMode.MoveUp
  # ...
  if ego.tactical_mode == TacticalMode.MoveUp:
    if in_interval(track_map.altitude(ego.track_mode)-ego.z, -1, 1):
      next.tactical_mode = TacticalMode.Normal 
      next.track_mode = track_map.h(ego.track_mode, ego.tactical_mode, TacticalMode.Normal)
  # ...
  return next
\end{minted}
\caption{\small Decision Logic Code Snippet from drone$\_$controller.py.}
\label{fig:controller}\label{code:guard_reset}
\end{figure}
\paragraph{Maps and sensors.} 
The map of a scenario specifies the tracks that the agents can follow. Besides creating from scratch, \ourtool provides functions that automatically generates  map objects from OpenDRIVE~\cite{noauthor_asam_nodate} files. The {\em sensor\/} function defines which variables from an agent are visible to other agents.
The default sensor function allows all agents to see all variables; we discuss how the sensor function can be modified to include bounded noise in Section~\ref{sec:exp:examples}. 
A map, a sensor and a collection of (compatible) agents together define a scenario object (Fig.~\ref{fig:scenario-code}). In the first few lines the drone agents are created, initialized, and added to the scenario object. 
A scenario can have heterogeneous agents with different decision logics. 
%
\begin{figure}
    \centering
    \scriptsize
\begin{minted}{python}
scenario = Scenario()
drone_red = DroneAgent('drone_red', file_name='drone_controller.py')
drone_red.set_initial([init_l_1, init_u_1],(CraftMode.Normal, TrackMode.T1))
scenario.add_agent(drone_red)
drone_blue = DroneAgent('drone_blue', file_name='drone_controller.py')
scenario.add_agent(drone_blue)
# ...
scenario.set_map(M6())
scenario.set_sensor(BaseSensor())
# ...
#traces = scenario.simulate(40, time_step)
traces = scenario.verify(40, time_step)
\end{minted}    \caption{Scenario Specification File Snippet}
    \label{fig:scenario-code}
\end{figure}
\paragraph{Simulation and reachability.} 
Once a scenario is defined, \ourtool's \simulate function can generate simulation(s) of the system, which can be stored and plotted. As shown earlier in Fig.~\ref{fig:ndsim}, a simulation from a single initial state explores all possible branches that can be generated by the decision logics of the  interacting agents, upto a specified time horizon. %
%
\ourtool can  verify the safety assertions of a scenario with the \verify function.  It computes the  over-approximations of the {\em reachable sets\/} for  each agent, and \yangge{checks} these against the predicates defined by the assertions. 
%
Fig.~\ref{fig:demo_graph} visualizes the result of such a computation performed using the \texttt{verify} function. In this example, the safety condition is violated when the blue drone moves downward to avoid the red drone. The other branches of the scenario are proved to be safe. 
%
%
The \simulate and \verify functions save a copy of the resulting execution tree, which can be loaded and traversed to analyze the sequences modes and states that leads to safety violations. 
\paragraph{Building advanced functions.} 
 \ourtool library provides  powerful functions to implement  advanced verification algorithms. Section \ref{sec:incremental-verification} explains how users can modify the basic reachability algorithm to save computing time in repeated verification of  similar scenarios using caching and incremental verification. \ourtool also makes it convenient to \yangge{plug in} different reachability subroutines. 
The experiments in Section~\ref{sec:exp:examples} \yangge{show} how the sensor can be useful when modeling realistic inputs to the controllers and Section~\ref{sec:uncertainty} describes verification when dynamics have uncertainty. 
\section{Scenarios in Verse}
\label{sec:scenarios}


A {\em scenario\/} in \ourtool is specified by a  map, a collection of agents in that map, and a sensor function that defines the part of each agent that is visible to other agents. We describe  these components below  and in  Section~\ref{sec:scenariotoHA} we will discuss how they formally define a hybrid system. 
\subsection{Tracks, modes, and maps} 
A {\em workspace\/} $\workspace$ is \yangge{an} Euclidean space in which the agents reside (For example, a compact subset of $\reals^2$ or $\reals^3$). 
%
An agent's continuous dynamics makes it roughly follow  certain continuous curves in $\workspace$,  called {\em tracks}, and occasionally the agent's decision logic changes the track. 
Formally, a {\em track\/} is simply a continuous function $\mapcurve: [0,1]\rightarrow \workspace$, but not all such functions are valid tracks. A map $\map$  defines the  set of  tracks $\mapsetcurve_\map$ it permits. 
In a highway map, some   tracks will be aligned along the lanes while others will correspond to merges and exits. 
%

We assume that an agent's decision logic does not depend on exactly which of the infinitely many tracks it is following, but instead, it depends only on  which type of track it is following or the {\em track mode}. 
In the example in Section~\ref{sec:example}, the track modes  are \texttt{T0}, \texttt{T1}, \texttt{M01}, etc. Every (blue) track for transitioning from point on T0 to the corresponding point on T1 has track mode \texttt{M01}.
A map has  a finite set of track modes $\maplabel_\map$,  a labeling function 
$\mapvert_\map: \mapsetcurve_\map \rightarrow \maplabel_\map$ that gives the track mode for a track. 
It also has a  mapping 
 $\maplabcurve_\map: \workspace \times \maplabel_\map \rightarrow \mapsetcurve_\map$ that gives a specific track from a track mode and a specific position in the workspace.  

Finally, a Verse agent's decision logic can change its internal mode or {\em tactical mode\/} (E.g., \texttt{Normal} to \texttt{MoveUp}). 
When an agent changes its tactical mode, it may also update the track \yangge{it is} following and this is encoded in another  function:  
$\maplablab_\map: \maplabel_\map \times \strategy \times \strategy \rightarrow \maplabel_\map$ which takes the  current track mode, the current and the next tactical mode, and generates the new track mode the agent should follow. For example, when the tactical mode of a drone changes from \texttt{Normal} to \texttt{MoveUP} while it  is on \texttt{T1}, \yangge{this map function $\maplablab_\map(\texttt{T1}, \texttt{Normal}, \texttt{MoveUp}) = \texttt{M10}$ informs that the agent should follow a track with mode $\texttt{M10}$}. 
These sets and functions together define a \ourtool map object $\map = \langle \maplabel_\map, \mapvert_\map,\maplabcurve_\map, \maplablab_\map \rangle$. We will drop the subscript $\map$, when the map being used is clear from context. 
%

\subsection{Agents}\label{sec:agents}

A Verse {\em agent\/} is defined by modes and continuous state variables, a decision logic that defines (possibly nondeterministic)  discrete transitions, and a flow function that defines continuous evolution. 
An agent $\agent$  is {\em compatible\/} with a  map $\map$ if the agent's tactical modes $\strategy$ are a subset of the allowed input tactical modes for $\maplablab$.
This makes it  possible to instantiate the same agent on different compatible maps. 
 The {\em mode space\/} for an agent instantiated on map $\map$ is the set $\Dagent = \maplabel \times \strategy$, where  $\maplabel$ is the set of  track modes in $\map$  and  $\strategy$ is the set of tactical modes of the agent. 
The {\em continuous state space\/} is $\Xagent = \workspace \times \Zagent$, where $\workspace$ is the workspace (of $\map$) and $\Zagent$ is the space of other continuous state variables.
The (full) {\em state space\/}  is the Cartesian product  $\agentfullstate = \Xagent \times \Dagent$. 
%
In the two-drone example in Section \ref{sec:example}, the continuous states variables  $px,py,pz,vx, vy, vz$ are the positions and velocities along the three axes of the workspace.
The modes are $\langle\texttt{Normal},\texttt{T1} \rangle$, $\langle\texttt{MoveUp},\texttt{M10}\rangle$, etc.
%

An {\em agent\/} $\agent$ in map $\map$ with $k-1$ other agents is defined by a tuple 
    $\agent = \langle \agentfullstate, \agentfullstate^0, \agentguard, \agentreset, \agentflow \rangle$, where
$\agentfullstate$ is the  state space, 
$\agentfullstate^0 \subseteq \agentfullstate$ is the set of initial states. 
%
%
The guard $\agentguard$ and reset $\agentreset$ functions jointly define the discrete transitions. For a pair of modes 
$\dagent, \dagent' \in \Dagent,$
$\agentguard(\dagent,\dagent') \subseteq \Xagent^k$ defines the condition under which a transition from 
$\dagent$ to $\dagent'$ is enabled. 
%
%
 The $\agentreset(\dagent, \dagent'):\Xagent^k \rightarrow \Xagent$ function specifies how the continuous states of the agent are updated when the mode switch happens. 
Both of these functions take as input the sensed continuous states of all the other $k-1$ agents in the scenario. Details about the sensor which transmits state information across agents is  discussed in Section~\ref{sec:sensor}.
The $\agentguard$  and the $\agentreset$ functions are actually not defined separately, but are extracted by the Verse parser from a block of structured Python code as shown in Fig.~\ref{code:guard_reset}. 
The discrete states in the \texttt{if} condition and the assignments define the source and destination of discrete transition. The \texttt{if} conditions involving continuous states define the guard for the transitions and the assignments of continuous states define the reset. 
Expressions with $\texttt{any}$ and $\texttt{all}$ functions are  unrolled to disjunctions and conjunctions according to the number of agents $k$.

For example, Lines 47-50 define transitions $\langle\texttt{MoveUp},\texttt{M10}\rangle$ to $\langle\texttt{Normal},\texttt{T0}\rangle$ and $\langle\texttt{MoveUp},\texttt{M21}\rangle$ to $\langle\texttt{Normal},\texttt{T1}\rangle$.
The change of track mode is given by the $\maplablab$ function. 
The guard for this transition comes from the \texttt{if} condition at Line 48. For example, $\agentguard(\langle\texttt{MoveUp},\texttt{M10}\rangle, \langle\texttt{Normal},\texttt{T0}\rangle) = \{\xagent\mid-1<\texttt{T0}.pz-\xagent.pz<1\}$ for $\xagent \in \Xagent$. Here continuous states remain unchanged after transition. 
The final component of the agent is the {\em flow\/}  function $\agentflow: \Xagent\times\Dagent\times\reals^{\geq 0} \rightarrow \Xagent$ which defines the continuous time evolution of the continuous state. For any initial condition  $\langle\xagent^0,\dagent^0\rangle \in Y$, $F(\xagent^0,\dagent^0)(\cdot)$ gives the continuous state of the agent as a function of time. In this paper, we use $F$ as a black-box function (see footnote~\ref{footnote:blackbox}). 


\subsection{Sensors and scenarios}
\label{sec:sensor}
For a scenario with $k$ agents, a {\em sensor\/}  function $\sensor:\agentfullstate^k \rightarrow \agentfullstate^{k}$ defines  the  continuous observables as a function of the continuous state. %
For simplifying exposition, in this paper we assume that observables have the  same type as the continuous state $Y$,  and that each agent $i$ is observed by all other agents identically. 
This simple, overtly transparent sensor model,  still allows us to write realistic agents that  only use information about  nearby agents. 
In a highway scenario, the observable part of agent $j$ to another agent $i$ may be the relative distance $y_{j} = x_j - x_i$, and vice versa, which can be computed as a function of the continuous state variables $x_j$ and $x_i$.
%
A different sensor function which gives  nondeterministic noisy observations, appears in Section~\ref{sec:exp:examples}.

A Verse {\em scenario} $\scenario$ is defined by (a) a map $\map$,
(b) a collection of $k$ agent instances $\{\agent_1 ... \agent_k\}$ that are compatible with $\map$, and  (c) a sensor $\sensor$ for the $k$ agents. 
%
Since all the agents are instantiated on the same compatible map $\map$, they share the same workspace. Currently, we require agents to have  identical state spaces, i.e., $\agentfullstate_i = \agentfullstate_j$, but they can have different decision logics  and different continuous dynamics.
\section{\ourtool scenario to hybrid verification}
\label{sec:scenariotoHA}

In this section, we define the underlying hybrid system   $\ha(\scenario)$, that a \ourtool scenario $\scenario$ specifies. The verification questions that \ourtool is equipped to answer are stated in terms of the behaviors or {\em executions\/} of $\ha(\scenario)$. \ourtool's notion of a hybrid automaton is close to that in Definition~5 of~\cite{dryvr}. As usual, the automaton has discrete and continuous states and discrete transitions defined by guards and resets. The only uncommon aspect in~\cite{dryvr}   is that the continuous flows may be defined by a black-box simulator functions, instead of white-box analytical models (see footnote~\ref{footnote:blackbox}). 

Given a scenario with $k$ agents  $\scenario = \langle \map, \{\agent_1,...\agent_k\}, \sensor, \strategy \rangle$, the corresponding hybrid automaton $\ha(\scenario) = \langle \Xha, \Xha^0, \Dha,  \Dha^0, \haguard, \hareset, \tl \rangle$, where
\begin{enumerate}
    \item $\Xha := \prod_i \Xagent_i$ is the  {\em continuous state space\/}. An element $\xha \in \Xha$ is called a {\em state}.
     $\Xha^0 := \prod_i \Xagent_i^0 \subseteq \Xha$ is the set of {\em initial continuous states.\/}
    \item $\Dha := \prod_i \Dagent_i$ is the {\em mode space\/}. An element $\dha \in \Dha$ is called a {\em mode}. 
     $\Dha^{0} := \prod_i \Dagent_i^0 \subseteq \Dha$ is the finite set of {\em initial modes}.
    \item For a mode pair  $\dha, \dha' \in \Dha$, $\haguard(\dha, \dha') \subseteq \Xha$ defines the continuous states from which a transition from $\dha$ to $ \dha'$ is enabled. A state
     $\xha\in \haguard(\dha, \dha')$ iff there exists an agent $i \in \{1,...,k\}$, such that $\xha_i \in\agentguard_i(\dha_i, \dha_i')$ and $\dha_j=\dha_j'$ for $j\neq i$.
%
    \item For a mode pair $\dha, \dha' \in \Dha$, $\hareset(\dha, \dha'):\Xha \rightarrow \Xha$ defines the change of continuous states after a transition from $\dha$ to $\dha'$. For a continuous state $\xha \in \Xha$, $\hareset(\dha, \dha')(\xha) = \agentreset_i(\dha_i,  \dha_i')(\xha) \text{ if } \xha \in \agentguard_i(\dha_i, \dha_i')$, otherwise $= \xha_i$.
    \item $\tl$ is a set of pairs $\langle \Tha, \dha\rangle$, where the {\em trajectory\/} $\Tha:[0,T] \rightarrow \Xha$ describes the evolution of continuous states in mode $\dha \in \Dha$. Given $\dha\in \Dha, \xha^0\in\Xha$, $\Tha$ should satisfy $\forall t\in\reals^{\geq 0}, \Tha_i(t) = F_i(\xha^0_i, \dha_i)(t)$.

\end{enumerate}

\begin{proposition}
If  a scenario with $k$ agents $\scenario = \langle \map, \{\agent_1,...,\agent_k\}, \sensor, \strategy \rangle$, satisfies the following: 
(1)  map and the agents are compatible,
(2) all agents have identical sets of states and modes, $\agentfullstate_i = \agentfullstate_j$, and (3) agent states match the input type of $\sensor$, then  $\ha(\scenario)$ is a valid hybrid system. 
\end{proposition}

We denote by $\Tha.fstate$, $\Tha.lstate$, and $\Tha.ltime$ the initial state $\Tha(0)$, the last state $\Tha(T)$, and $\Tha.ltime = T$. 
%
For a sampling parameter $\delta >0$ and a length $m$, a $\delta$-{\em execution\/} of a hybrid automaton $\ha = \ha(\scenario)$ is a sequence of $m$ labeled trajectories $\exec := \langle \Tha^0, \dha^0\rangle,...,\langle\Tha^{m-1}, \dha^{m-1}\rangle$, such that 
(1) $\Tha^0.fstate \in \Xha^0, \dha^0 \in \Dha^0$,
(2) For each $i \in \{1,...,m-1\}$, $\Tha_i.lstate\in\haguard(\dha^i, \dha^{i+1})$ and $\Tha^{i+1}.fstate = \hareset(\dha^i, \dha^{i+1})(\Tha^i.lstate)$, and 
(3) For each $i \in \{1,...,m-1\}$, $\Tha^i.ltime = \delta$ for $i\neq m-1$ and $\Tha^i.ltime \leq \delta$ for $i=m-1$. 

We define the first and last state of an execution $\exec = \langle \Tha^0, \dha^0\rangle,...,\langle\Tha^{m-1}, \dha^{m-1}\rangle$ as $\exec.fstate = \Tha^0.fstate$, $\exec.lstate=\Tha^{m-1}.lstate$ and the first and last mode as $\exec.fmode = \dha^0$ and $\exec.lmode = \dha^{m-1}$. The set of reachable states is defined by $\reach{} := \{\exec.lstate \mid \exec\text{ is an execution of }\ha\}$. In addition, we denote the reachable states in a specific mode $\dha \in \transvert$ as $\reach{}(\dha)$ and $\reach{}(T)$ to be the set of reachable states  at time $T$.
Similarly,  denoting the unsafe states for mode $\dha$ as $\unsafe(\dha)$, the safety verification problem for $\ha$ can be solved by checking whether  $\forall \dha \in \Dha$, $\reach{}(\dha) \cap \unsafe(\dha) = \emptyset$.
Next, we discuss \ourtool functions for verification via reachability. 

%

\section{Building verification algorithms in \ourtool}
\label{sec:verification_algorithm}

The Verse library comes \yangge{with} several built-in verification algorithms, and it  provides functions that users can be use to implement powerful new algorithms. We first describe the basic building blocks and in Sections~\ref{sec:incremental-verification} and~\ref{sec:exp:inc} we discuss more advanced use cases and algorithms.

\subsection{Reachability analysis}
\label{sec:bfs}
Consider  a scenario $\scenario$ with $k$ agents and the corresponding hybrid automaton $\ha(\scenario)$. For a pair of modes, $\dha,\dha'$ the standard discrete $\post_{\dha,\dha'}:\Xha \rightarrow \Xha$ and continuous $\post_{\dha,\delta}:\Xha \rightarrow \Xha$ operators are defined as follows:
For any state $\xha, \xha' \in \Xha$,
$\post_{\dha,\dha'}(\xha) = \xha'$
iff $\xha \in \haguard(\dha,\dha')$ and $\xha'=\hareset(\dha,\dha')(\xha)$;
and, $\post_{\dha,\delta}(\xha) = \xha'$ iff $\forall i\in {1,...,k}$, $\xha_i'=\agentflow_i(\xha_i, \dha_i, \delta)$. 
These operators are also lifted to sets of states in the usual way. \ourtool provides \postcont to compute $\post_{\dha,\delta}$ and \postdisc to compute $\post_{\dha,\dha'}$. Instead of computing the exact post, \postcont and \postdisc compute over-approximations using improved implementations of the algorithms in~\cite{dryvr}.
Verse's \verify function implements a  reachability analysis algorithm  using these post operators (Algorithm~\ref{alg:veri}). This algorithm constructs an execution tree $\exectree = \langle\treevert,\treeedge \rangle$ up to depth $m$ in breadth first order. Each vertex $\langle \Xhaset, \dha \rangle \in \treevert$ is a pair of \yangge{a set of states and a mode}. The root is $\langle \Xha^0, \dha^0\rangle$. There is an edge from 
 $\langle \Xhaset, \dha \rangle$ to $\langle \Xhaset', \dha'\rangle$, iff \yangge{$\Xhaset' = \post_{\dha',\delta}(\post_{\dha,\dha'}(\Xhaset))$}. 
%



\begin{algorithm}
\caption{ }
\label{alg:veri}
\begin{algorithmic}[1]
    \Function{\algveri}{$\langle\Xha^0,\dha^0\rangle, \ha,\treedepth, \delta$}
    \State $root \leftarrow \langle\Xha^0,\dha^0\rangle$ 
    \State $depth \leftarrow 0$
    \State $queue \leftarrow[\langle\langle\Xha^0, \dha^0 \rangle,depth\rangle]$ \label{alg:veri:root}
    \While{$queue \neq \emptyset$ and $depth\leq \treedepth$} 
        \State $\langle \sd{},depth \rangle\leftarrow queue.\texttt{dequeue}()$ \label{alg:veri:order}
        \For{$\dha'$, s.t. $\haguard(\dha,\dha')$}
            \If{$\Xhaset \cap \haguard(\dha,\dha')\neq \emptyset$} \label{alg:veri:guard_check}
                \State $\sdp{}\leftarrow\langle \postcont(\postdisc(\Xhaset,\dha,\dha'),\dha', \delta), \dha'\rangle$; 
                \State $\sd{}.\texttt{addChild}(\langle\Xhaset',\dha'\rangle)$ \label{alg:veri:postdd}
                \State $queue.\texttt{add}(\langle \sdp{}, depth+1 \rangle)$ \label{alg:veri:next}
            \EndIf
        \EndFor 
    \EndWhile
    \State \Return $root$
    \EndFunction
\end{algorithmic}
\end{algorithm}




\begin{proposition}
\label{prop:veri}
For hybrid automaton $\ha = \ha(\scenario)$, let $\exectree=\langle \treevert, \treeedge \rangle$ be the execution tree with depth $T$ constructed by {\bf \algveri}, then
for each level $t\in \{0,...,T\}$, 
$\reach{}(\delta t) = V(t)$,
where $\treevert(t)$ is the union of the states in the vertices at depth $t$.
\end{proposition}
Proposition \ref{prop:veri} holds when the post computations are exact. However, typically, we rely on techniques that compute the over-approximations of the actual post, in which, $\treevert(t)$ over-approximates $\reach{}(\delta t)$. 
Currently, \ourtool  implements only  bounded time reachability, however, basic unbounded time analysis with fixed-point checks could be added following~\cite{dryvr,sibai2021scenechecker}.

\subsection{Incremental Verification}
\label{sec:incremental-verification}

During the design-analysis process, users  perform many simulation and verification runs on slightly tweaked scenarios. Can we do better than starting each verification run from scratch? In \ourtool, we have implemented an incremental analysis algorithm that improves the performance of \texttt{simulate} and \verify by reusing data from previous verification runs. This algorithm also illustrates how \ourtool can be used to implement different algorithms. 

Consider two hybrid automata $\ha_i = \ha(\scenario_i),$ $i \in \{1,2\}$ that only differ in the discrete transitions. That is,
(1) $\Xha_2 = \Xha_1$, 
(2) $\Dha_2 = \Dha_1$,
and 
(3) $\tl_2 = \tl_1$, while the initial conditions, the guards, and the resets are slightly  different\footnote{Note that in this section subscripts index  different hybrid automata, instead of agents within the same  automaton (as we did in Sections~\ref{sec:scenarios} and \ref{sec:scenariotoHA}).}.  
$\scenario_1$ and $\scenario_2$ have the same sensors, maps, and agent flow functions. Let $\exectree_1=\langle\treevert_1,\treeedge_1\rangle$ and $\exectree_2=\langle\treevert_2,\treeedge_2\rangle$ be the execution trees for  $\ha_1$ and $\ha_2$. 
Our idea of incremental verification is to reuse some of the computations in constructing the tree for $\ha_1$ in computing the same for $\ha_2$. 

Recall that in \algveri, expanding each  vertex $\langle\Xhaset_1,\dha_1\rangle$ of $\exectree_1$
with a  possible mode  involves a guard check, a computation of $\post_{\dha,\dha'}$, and $\post_{\dha,\delta}$. 
The \alginc algorithm avoids performing these computations while constructing $\exectree_2$ by reusing those computations from $\exectree_1$, if possible. To this end, \alginc is endowed with two caches: $\guardcache$
stores information about discrete transitions and $\flowcache$ stores information about continuous flows.
The key to $\guardcache$ is a vertex $\langle\Xhaset_2,\dha_2\rangle\in \treevert_2$, a guard  $\haguard_2(\dha_2,\dha_2')$ and a reset $\hareset_2(\dha_2,\dha_2')$ for a mode pair. 
%
%
The retrieved data from $\guardcache$ will be a pair $\langle s, v'\rangle$ where $s\in\{\sat,\unsat,\unknown\}$ is the guard checking result for $\Xhaset_2$ and $\haguard_2(\dha,\dha')$ and $v'=\langle\post_{\dha,\dha'}(\Xhaset_2),\dha'\rangle$ is the vertex after applying post. 
A cache hit can happen if there exists an entry in the cache with key match exactly with $\langle \langle\Xhaset_2,\dha_2 \rangle, \haguard_2(\dha_2,\dha_2'), \hareset_2(\dha_2,\dha_2')\rangle$.
%
%
If $\langle \sd{2}, \haguard_2(\dha_2,\dha_2'), \hareset_2(\dha_2,\dha_2')\rangle\in \guardcache$
and $s=\sat$, then we know from $\guardcache$ that $\Xhaset_2$ satisfies $\haguard_2(\dha_2,\dha_2')$ 
and the post of $\Xhaset$ can be retrieved from the cache $v'=\langle\post_{\dha,\dha'}(\Xhaset_2), \dha_2'\rangle$. If $\langle \sd{2}, \haguard_2(\dha_2,\dha_2'), \hareset_2(\dha_2,\dha_2')\rangle\in \guardcache$ but $s=\unsat$, then $\Xhaset_2$ does not satisfy the guard $\haguard_2(\dha_2,\dha_2')$ and $v'=\none$. If $\langle \sd{2}, \haguard_2(\dha_2,\dha_2'), \hareset_2(\dha_2,\dha_2')\rangle\notin \guardcache$, then $s=\unknown$, $v'=\none$ and the guard checking and post computation from $\Xhaset_2$ will need to happen afresh for $\ha_2$. 

The second cache constructed is $\flowcache$. The key to the cache will be a vertex $\langle\Xhaset,\dha \rangle$. A cache hit can happen if there exist an entry in the cache $\langle\Xhaset',\dha'\rangle$ such that $\Xhaset\subseteq \Xhaset'$ and $\dha=\dha'$ and the retrieved value from the cache will be $\Xhaset''=\post_{\dha,\delta}(\Xhaset')\supseteq \post_{\dha,\delta}(\Xhaset)$.



Similar to \algveri, for each vertex in the tree, \alginc expands all possible mode transitions (Line \ref{alg:inc:exp_trans}).The full algorithm is shown in Algorithm~\ref{alg:inc} for completeness and we skip the \yangge{detailed} description owing to space limitations. \alginc checks $\guardcache$ or $\flowcache$ before every $\post$ computation to retrieve and reuse computations when possible. The caches can save information from any number of previous executions, so \alginc can be even more efficient than \algveri when running many consecutive verification runs.

\begin{proposition}
\label{prop:inc}
Given the caches $\guardcache$ and $\flowcache$ constructed while constructing $\exectree_1$ for $\ha_1$,
the tree $\exectree_2$ constructed by \alginc is sound. That is
\[
\reach{2}(\delta t)\subseteq \treevert_{2}(t), \forall t\in \{0,...,T\}
\]
\end{proposition}
\begin{proof}
The proof relies on \alginc's design ensuring that the cache data is always an over approximation of the actual post computations. 
Consider $\guardcache$ and $\flowcache$ constructed while constructing $\exectree_1$ and the tree growth algorithm \alginc from a vertex $\langle \Xhaset_2, \dha_2\rangle \in \exectree_2$ for tree for $\ha_2$. When no cache hit happens, \alginc works exactly the same as \algveri. 
If $\langle \langle \Xhaset_2, \dha_2\rangle,\haguard_2(\dha_2,\dha_2'),\\\hareset_2(\dha_2,{\dha_2}') \rangle\in \guardcache$, there exists an entry in $\guardcache$ with key $\langle \langle \Xhaset_1, \dha_1\rangle,\haguard_1(\dha_1,{\dha_1}'),\\\hareset_1(\dha_1,{\dha_1}')\rangle$ from $\exectree_1$ that matches exactly with it and the output from cache will be $v'=\langle\post_{\dha,\dha'}(\Xhaset_1),{\dha_1}' \rangle=\langle\post_{\dha,\dha'}(\Xhaset_2),{\dha_2}' \rangle$.
Similarly for $\flowcache$, given $\langle \Xhaset_2,\dha_2\rangle$, a cache hit can happen only when there exists an entry $\langle \Xhaset_1,\dha_1\rangle$ in cache such that $\dha_1=\dha_2$ and $\Xhaset_1 \supseteq \Xhaset_2$. Therefore, the output from cache will be $\post_{\dha,\delta}(\Xhaset_1) \supseteq \post_{\dha,\delta}(\Xhaset_2)$. 
\end{proof}


\begin{algorithm}[h]
\caption{ }
\label{alg:inc}
\begin{algorithmic}[1] 
	\State \textbf{Global} $\guardcache$, $\flowcache$ 
    \Function{$\algfcache$}{$\sd{}$} \label{alg:flow_cache}
        \If{$\langle\Xhaset,\dha\rangle \notin \flowcache$} \label{alg:inc:flow_cache_query1} $\flowcache(\langle\Xhaset,\dha\rangle) \leftarrow \postcont(\Xhaset,\dha,\delta)$ \label{alg:inc:flow_add} 
        \EndIf            
        \State \Return $\langle \flowcache(\langle\Xhaset,\dha\rangle), \dha\rangle$ \label{alg:inc:flow_cache_hit1}
    \EndFunction \hrulefill 
    \Function{\alginc}{$\langle\Xha^0,\dha^0\rangle,\ha,\treedepth$}
    \State $root \leftarrow \langle\Xha^0,\dha^0\rangle$; $depth \leftarrow 0$
    \State $queue \leftarrow [\langle\Xha^0,\dha^0,depth\rangle]$ \label{alg:inc:root} 
    \While{$queue\neq \emptyset$ and $depth \leq \treedepth$} 
        \State $\langle\sd{},depth\rangle\leftarrow queue.\texttt{dequeue}()$ \label{alg:inc:order}
        \For{$\dha'$, s.t. $\haguard(\dha,\dha')$} \label{alg:inc:exp_trans}
            \State $\langle s,\sdp{}\rangle \leftarrow \guardcache(\sd{}, \haguard(\dha,\dha'), \hareset(\dha,\dha'))$ \Comment{Read $\guardcache$}\label{alg:inc:guard_cache_hit}
            \If{$s\neq \unsat$}  \Comment{if $s=\unsat$ then continue to next $\dha'$}
                \If{s=\unknown}
                    \If{$\Xhaset \cap \haguard(\dha,\dha')\neq \emptyset$} \label{alg:inc:guard_check}
                        \State $\sdp{}\leftarrow\langle \postdisc(\Xhaset,\dha,\dha'), \dha'\rangle$ \label{alg:inc:postdd}
                        \State $\guardcache(\langle \sd{},\haguard(\dha,\dha') , \hareset(\dha,\dha')\rangle) \leftarrow \langle \sat,\sdp{}\rangle$ \label{alg:inc:guard_cache_add1} \Comment{Record in $\guardcache$}
                    \Else
                        \State $\guardcache(\langle \sd{},\haguard(\dha,\dha') , \hareset(\dha,\dha')\rangle) \leftarrow \langle \unsat,\none\rangle$ \label{alg:inc:guard_cache_add2} \Comment{Record in $\guardcache$}
                        \State \textbf{continue}
                    \EndIf
                \EndIf
                \State $\sdpp{}\leftarrow \algfcache(\sdp{})$ 
                \State $\sd{}.\texttt{addChild}(\sdpp{})$ \label{alg:inc:flowcache1}
                \State $queue.\texttt{add}(\langle \sdpp{},depth+1 \rangle)$ \label{alg:inc:queue}
            \EndIf
        \EndFor 
    \EndWhile
    \State \Return root
    \EndFunction
\end{algorithmic}
\end{algorithm}

\section{Experimental Evaluation}
\label{sec:exp}
We evaluate key features and algorithms  in  \ourtool through examples. 
We  consider two types of agents: 
a 4-d ground vehicle with bicycle dynamics and the Stanley controller~\cite{4282788} and a 6-d drone with a NN-controller~\cite{ivanov2019verisig}. 
%
Each of these agents can be fitted with one of two types of decision logic: (1) a collision avoidance logic (CA) by which the agent switches to a different available track when it nears another agent on its own track,  and (2) a simpler non-player vehicle logic (NPV) by which the agent does not react to other agents (and just follows its own track at constant speed). We denote the car agent with CA logic as agent C-CA, drone with  NPV as D-NPV, and so on.  
We use four 2-d maps ($\map1$-$4$) and two 3-d maps $\map5$-$6$. 
$\map1$ and $\map2$ have $3$ and $5$ parallel straight tracks, respectively. $\map3$ has $3$ parallel tracks with circular curve. $\map4$ is  imported from OpenDRIVE. $\map6$ is the figure-8 map used in Section~\ref{sec:example}. 

\subsection{Evaluation of core features}
\label{sec:exp:examples}
\paragraph{Safety analysis with multiple drones in a 3-d map.}
The first example is a scenario with two drones---D-CA agent (red) and D-NPV agent (blue)---in map $\map5$. The safety assertion requires agents to always separate by at least 1m. 
Fig.~\ref{fig:exp:simple_2} ({\em left}) shows the  computed reachable set, its projection on $x$-position, and on $z$ position. Since the agents are separated in space-time, the scenario is verified safe. These plots are generated using \ourtool's plotting functions. 

\begin{figure}[!htb]
        \centering
        \includegraphics[width=0.24\textwidth,height=1.3cm]{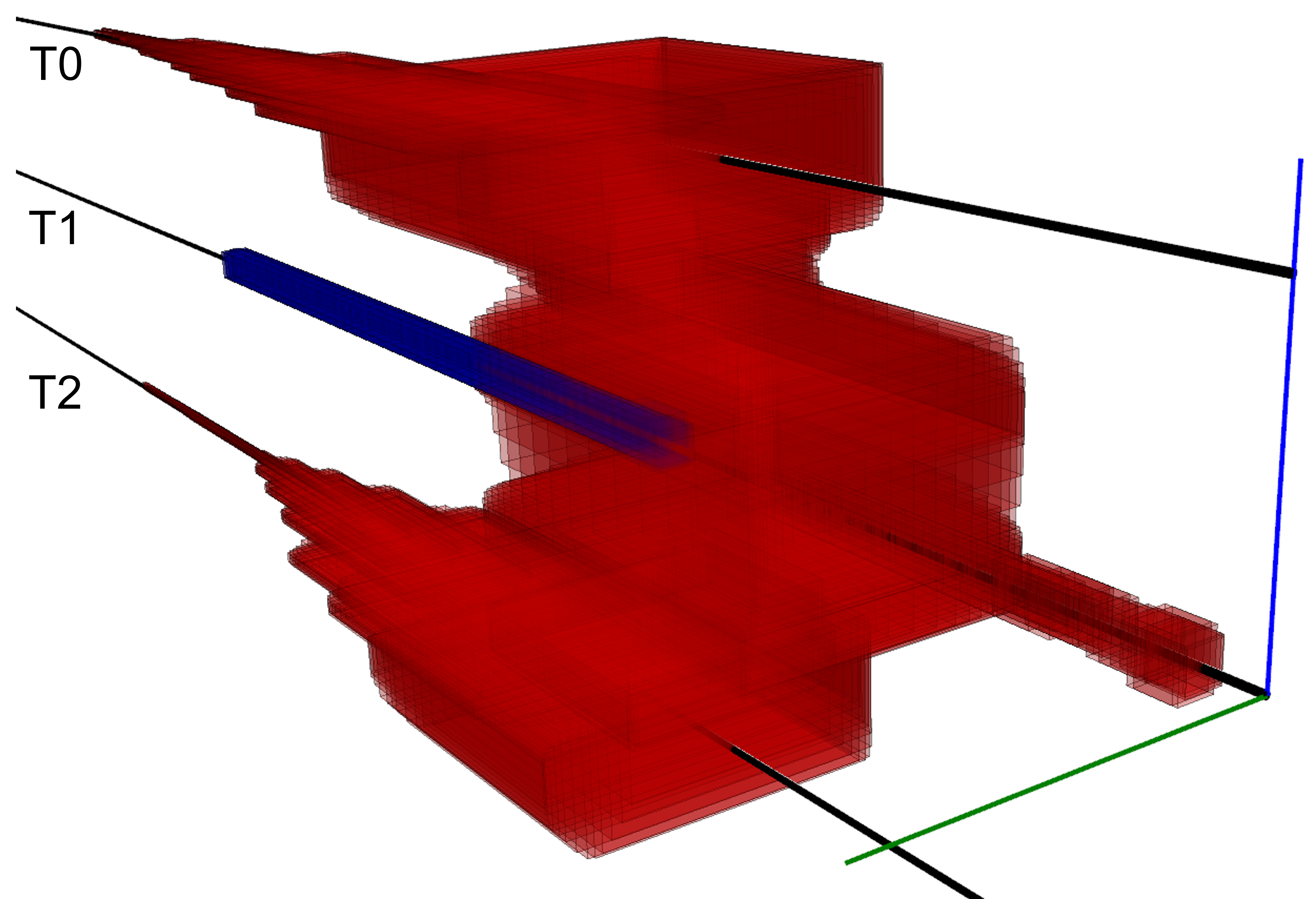}
        \includegraphics[width=0.24\textwidth,height=1.3cm]{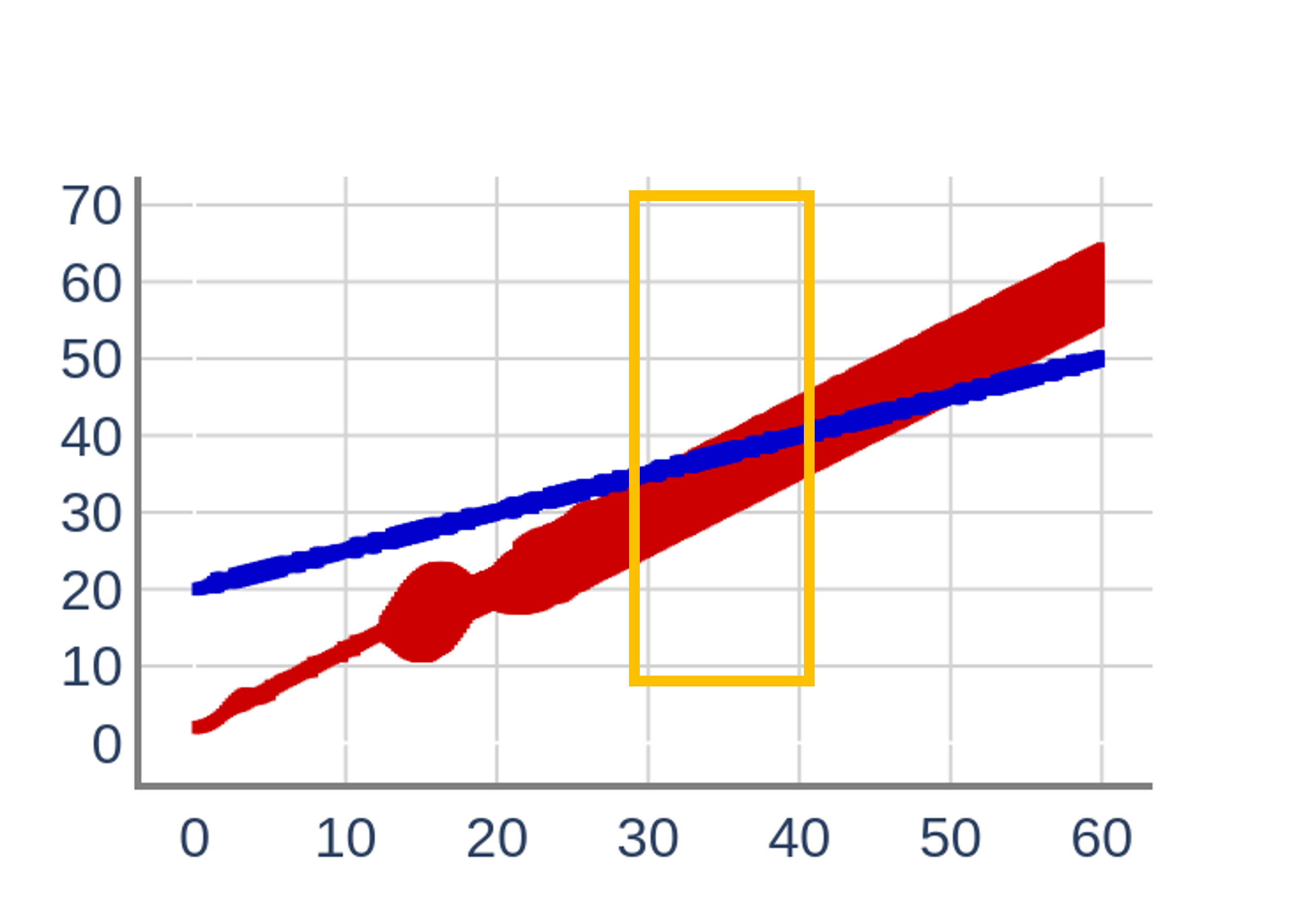}
        \includegraphics[width=0.24\textwidth,height=1.3cm]{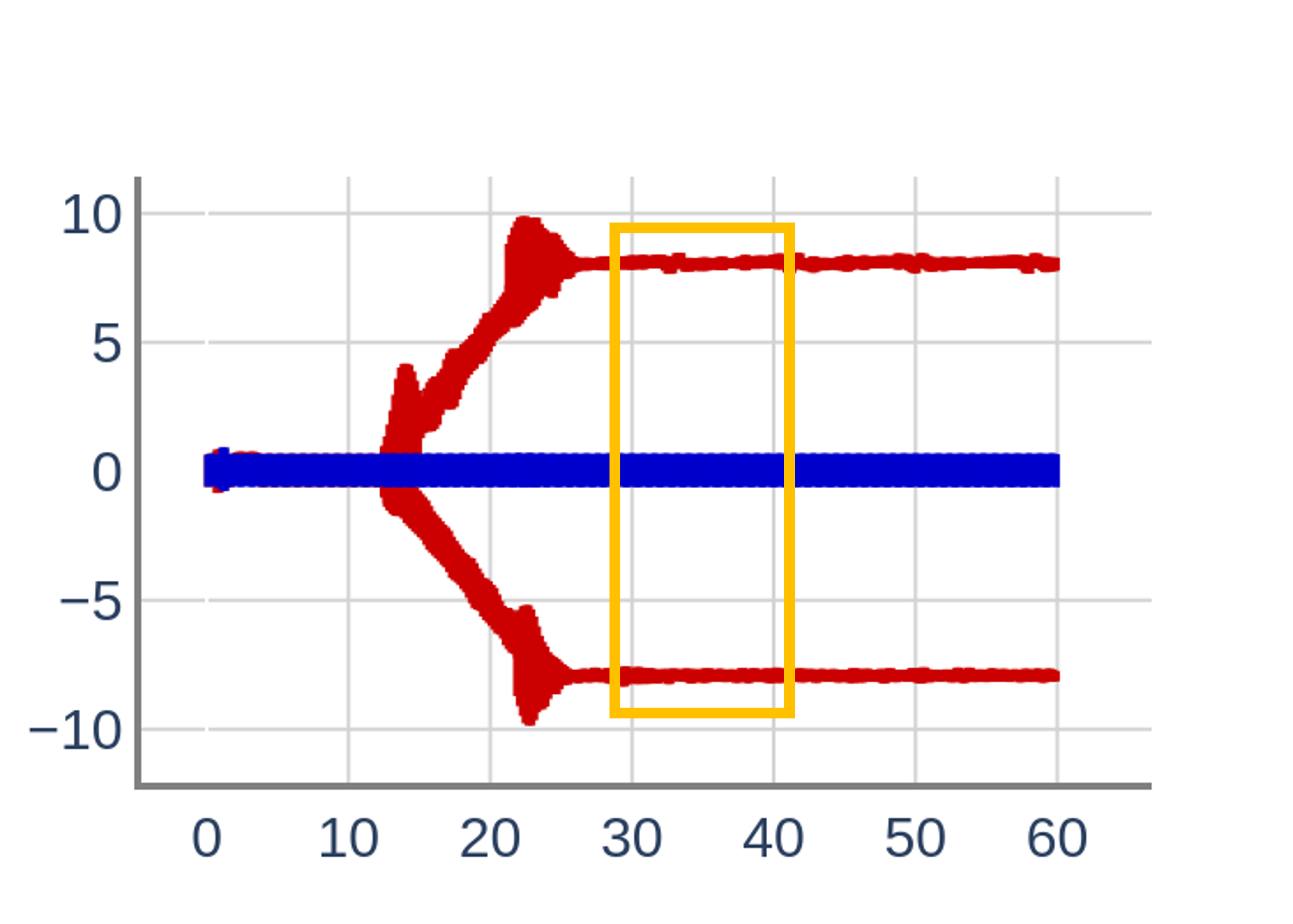}
        \includegraphics[width=0.24\textwidth,height=1.3cm]{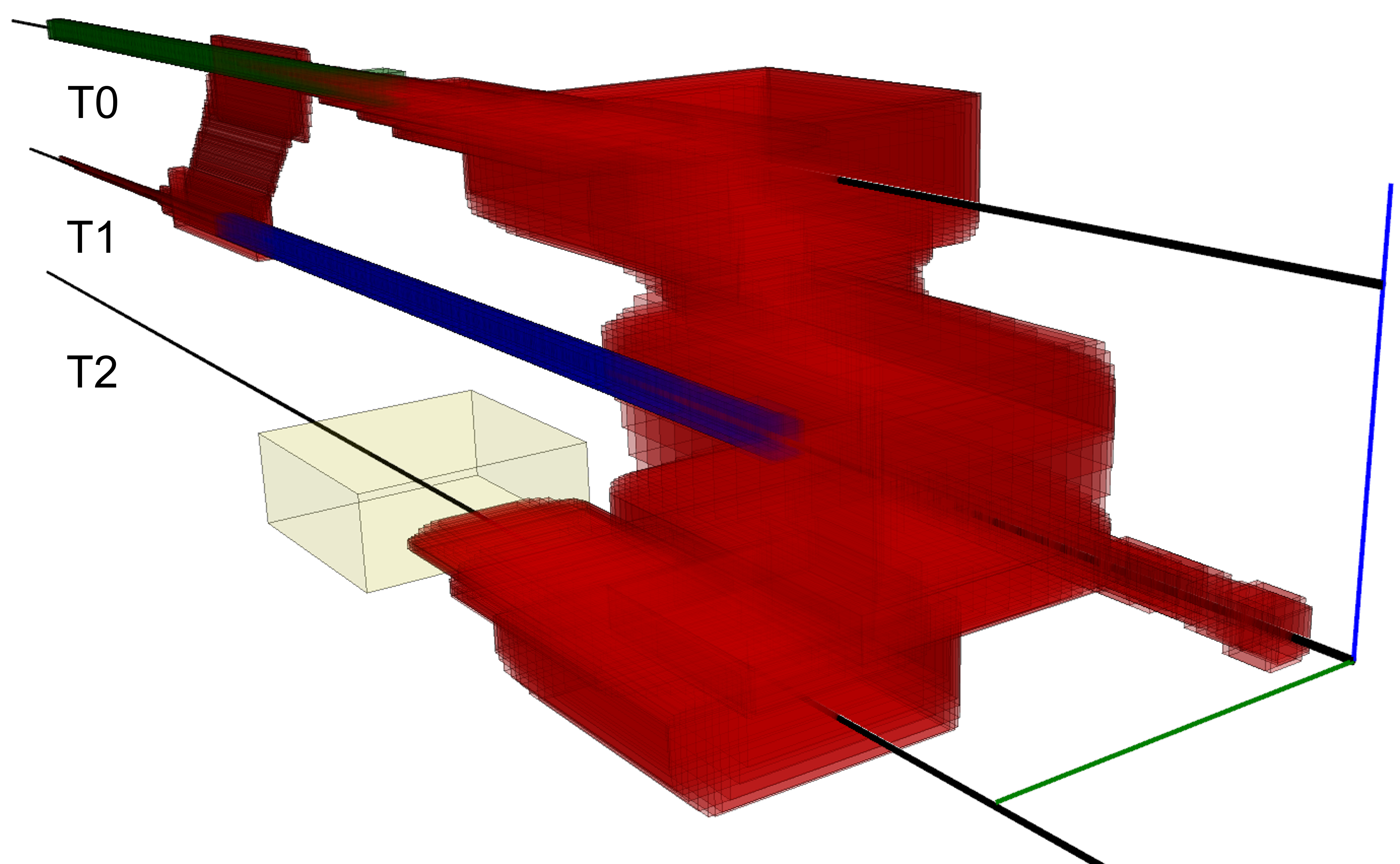}

        \caption{\small Left to right: (1) Computed reachtubes for a 2-drone scenario; (2) same reachtube projected on  x-dimension, and (3) on z-dimension. Since there is no overlap in space-time, no collision. (4) Reachtube for a 3-drone scenario, the red drone violates the safety condition by entering the unsafe region after moving downward.}\label{fig:exp:simple_2}
\end{figure}

\paragraph{Checking multiple safety assertions.} 
\ourtool supports multiple safety assertions specified using $\texttt{assert}$ statements. 
For example, the user can specify unsafe regions (Line 77-80) or check safe separation between agents (Line 81-83) as shown in Fig.~\ref{code:unsafe}.
We extend the two-drone scenario described in the previous paragraph by adding a second D-NPV and both safety assertions. The result is shown in the \yangge{rightmost} Fig.~\ref{fig:exp:simple_2}. In this scenario, D-CA violates the safety property by entering the unsafe region after moving downward to avoid collision.
The behavior of D-CA after moving upward is not influenced. There's no violation of safe separation between agents in this scenario. \ourtool allow \yangge{users} to extract set of reachable states and mode transitions that leads to a safety violation. 


\begin{figure}
    \centering
    \scriptsize
\begin{minted}{python}
    assert not (ego.x > 40 and ego.x<50 and\
        ego.y>-5 and ego.y<5 and ego.z > -10 and ego.z<-6), "Unsafe Region"
    assert not any(ego.x-other.x < 1 and ego.x-other.x >-1 and \
        ego.y-other.y < 1 and ego.y-other.y > -1 and \ 
        ego.z-other.z < 1 and ego.z-other.z > -1 \
        for other in others), "Safe Separation"
    return next
\end{minted}
\caption{\small Safety assertions for three drone scenario.}
\label{code:unsafe}
\end{figure}

\paragraph{Changing maps.} \ourtool allows users to easily create scenarios with different maps and port agents across compatible maps. We start with a  scenario with one C-CA agent (red) and two C-NPV agents (blue, green) in $\map1$. The safety assertion  is that the vehicles should be at least 1m apart in both $x$ and $y$-dimensions.  Fig.~\ref{fig:exp:map_safe} ({\em left\/}) shows the  verification result  and safety is not violated. 
However, if we switch to map $\map3$ by changing one line in the scenario definition, a reachability analysis shows that a safety violation can happen after C-CA merges left  Fig.~\ref{fig:exp:map_safe} ({\em center}). 
%
In addition, \ourtool allows importing map from OpenDRIVE \cite{noauthor_asam_nodate} format, which enables users to easily create scenarios with interesting maps. An example is shown in Fig. \ref{fig:exp:opendrive} in Appendix. 


\begin{figure}
    \centering
    \includegraphics[width=0.32\textwidth,height=2.0cm]{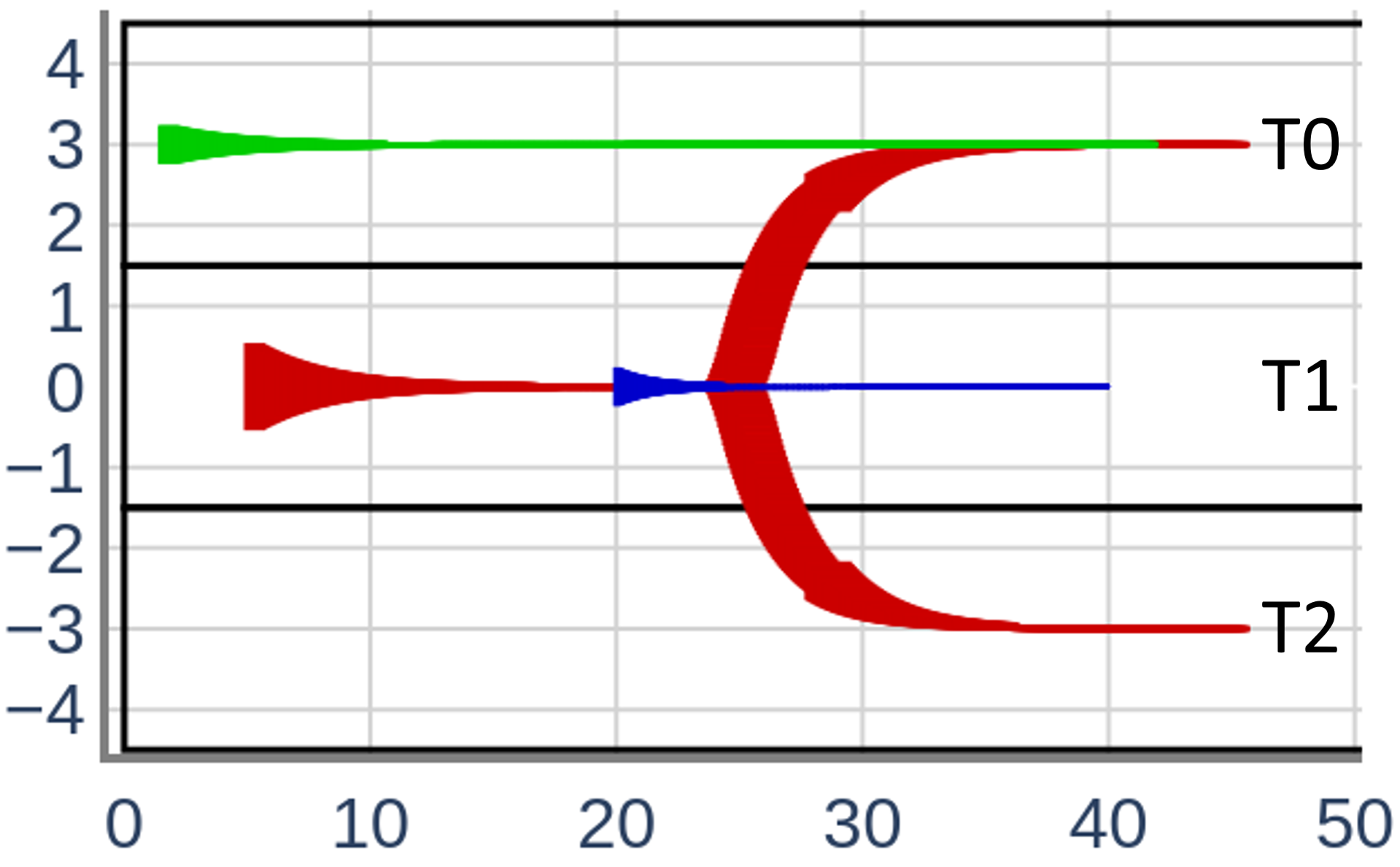}
    \includegraphics[width=0.32\textwidth,height=2.0cm]{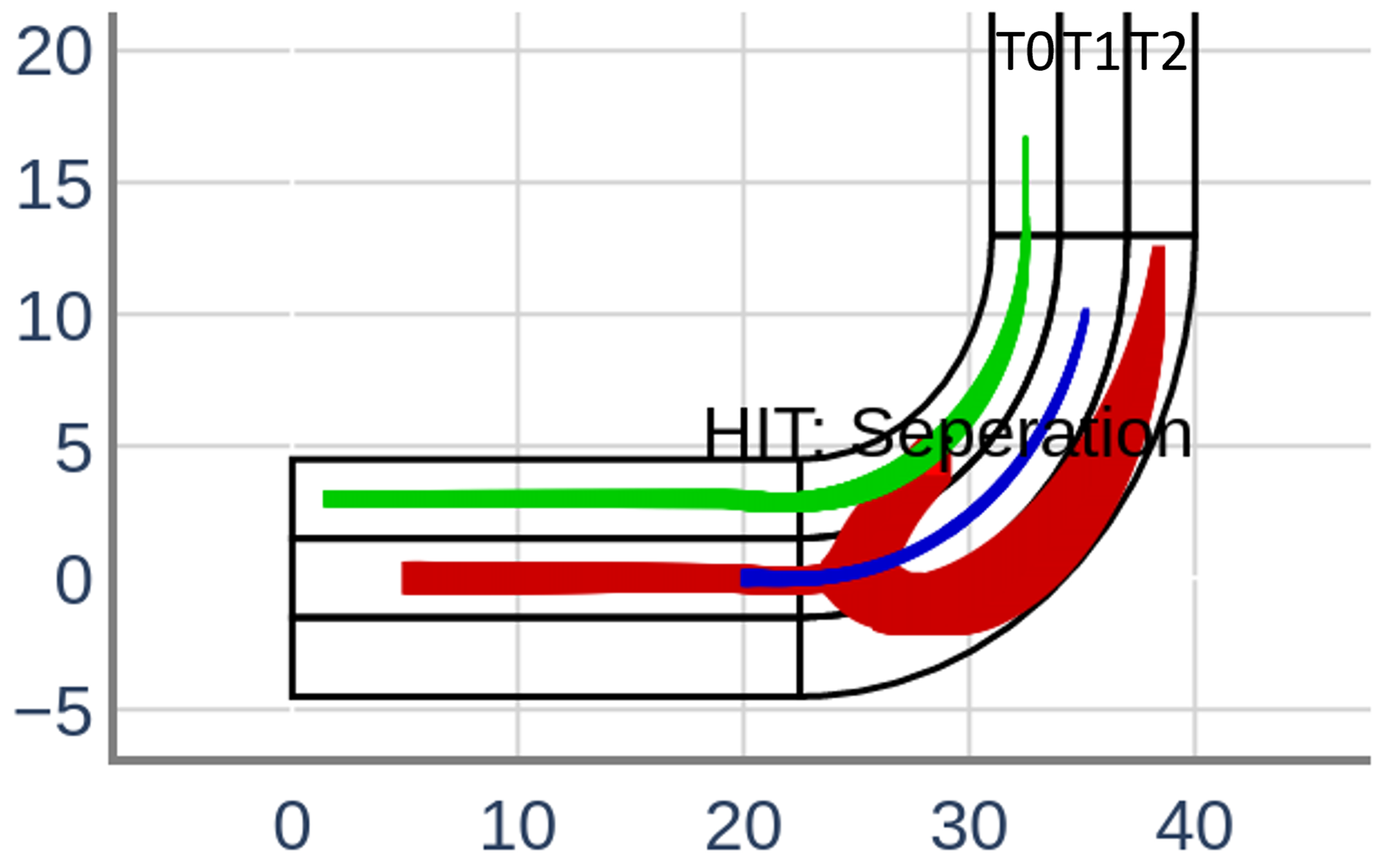}
    \includegraphics[width=0.32\textwidth,height=2.0cm]{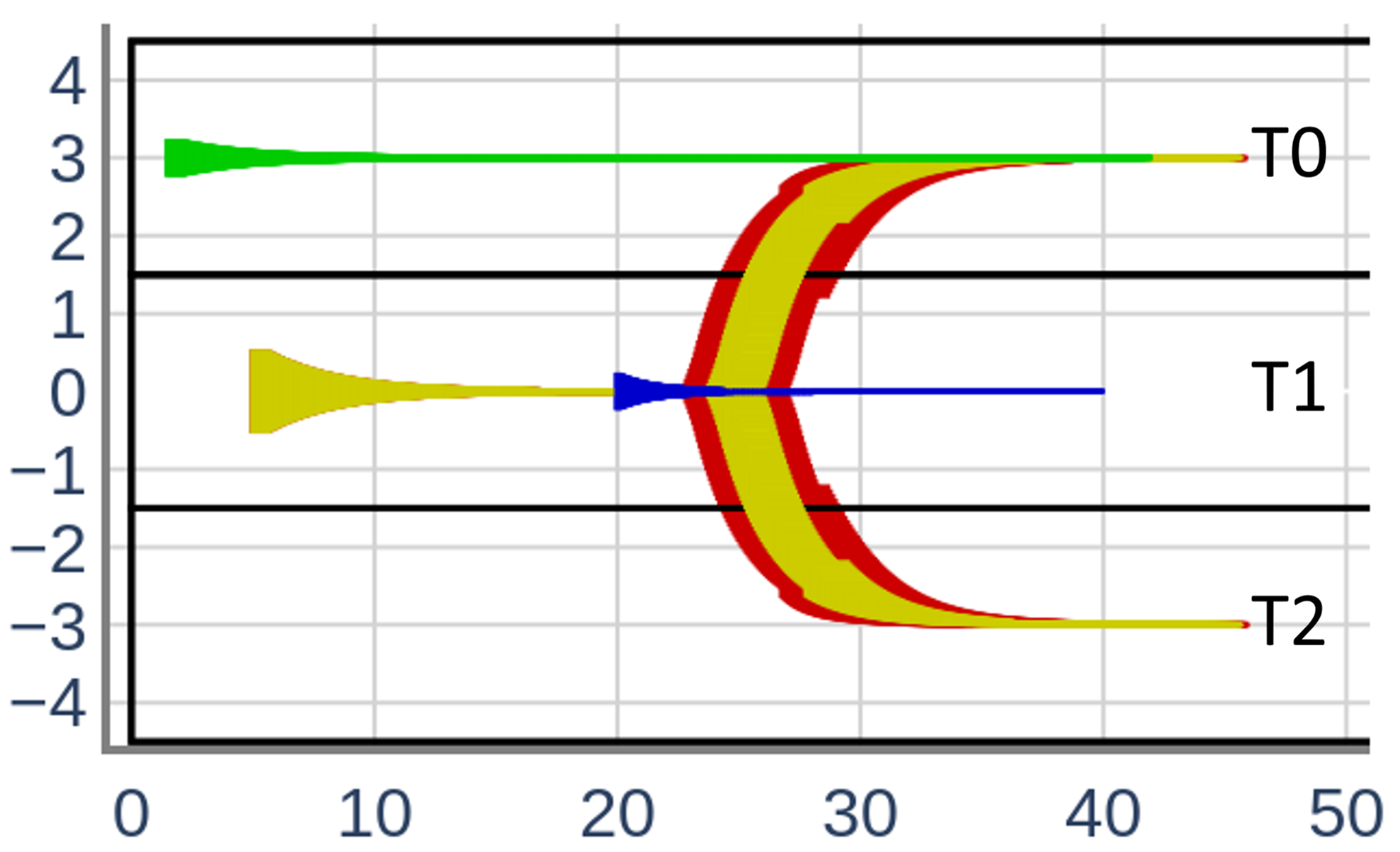}
    \caption{\small \textit{Left:} running the three car scenario on map with parallel straight lanes.  \textit{Center:} same scenario with a curved map. \textit{Right:} same scenario with a noisy sensor.} 
    \label{fig:exp:map_safe}
\end{figure}


\paragraph{Adding noisy sensors.} \ourtool supports the ability to verify scenarios with different sensor functions. For example,  the user can  create a noisy sensor function \yangge{that mimics} a realistic sensor with bounded noise. Such sensor functions are easily added to the scenario using the $\texttt{set\_sensor}$ \ourtool library function.

Fig.~\ref{fig:exp:map_safe} shows exactly the same three-car scenario as in the previous paragraph but with a noisy sensor, which adds $\pm0.5 m$ noise to the perceived position of all other vehicles. 
Since the sensed values of other agents only impacts the checking of the guards (and hence the transitions) of the ego agent, \ourtool internally bloats the reachable set of positions for the other agents by $\pm0.5$ while checking guards.
%
%
Compared with the behavior of the same agent with no sensor noise (shown in yellow in Fig~\ref{fig:exp:map_safe} ({\em right})), the sensor noise enlarges the region over which  the transition can happen, which results in enlarged  reachtubes for the red agent.  




\paragraph{Plugging in different reachability engines.}
With a little effort, \ourtool allows users to plug-in different reachability tools for the \postcont computation. The user will need to modify the interface of the reachability tool  so that given a set of initial states, a mode, and a non negative value $\delta$, the reachability tool can output the set of reachable states over a $\delta$-period represented by a set of timed hyperrectangles.
%
Currently, \ourtool implements computing $\postcont$ using DryVR~\cite{dryvr}, NeuReach~\cite{neuReach-SunM22} and Mixed Monotone Decomposition~\cite{9304391}. 
%
A scenario with two car agents in map $\map1$ verified using NeuReach and DryVR is shown in Fig.~\ref{fig:exp:neureach_post} in the Appendix.

Table~\ref{tab:exps}  summarizes the running time 
of verifying all the  examples in this section on 
a standard Intel Core i7-11700K @ 3.60GHz CPU desktop. 
As expected, the running times increase with the  number of discrete mode transition. However, for complicated scenario with $7$ agents and \yangge{$37$} transitions, the verification can still finish in under $6$ mins,  which suggests some level of scalability. 
The choice of reachability engine can also impact running time. For the same scenario in  rows 2,3 and 10,11, \ourtool with NeuReach\footnote{Run time for NeuReach includes training time.} as the reachability engine takes more time than using DryVR as the reachability engine. 

\begin{table}[!htp]
\centering
\caption{\small Runtime for verifying examples in Section~\ref{sec:exp:examples}. Columns are: number of agents (\#$\agent$), agent type ($\agent$), map used (Map), reachability engine used  ($\postcont$), sensor type (Noisy $\sensor$), number of mode transitions \#TR, and the total run time (Run time).}
\vspace{0.25cm}
\label{tab:exps}
\small
\begin{tabular}{| r | r | r | r | r | r | r | r |}\hline
\small
\#$\agent$&$\agent$ & Map & $\postcont$ & Noisy $\sensor$   &\#Tr& Run time (s)  \\\hline 
2&Q&$\map6$ &DryVR &No  &8 & 55.9\\ 
2&Q&$\map5$ &DryVR &No  &5&  18.7\\ 
2&Q&$\map5$ &NeuReach &No  &5&  1071.2\\ 
3&Q& $\map5$& DryVR &No  &7 & 39.6\\  
7& C& $\map2$ & DryVR &No  &37 &322.7 \\ 
3& C& $\map1$ & DryVR &No  &5& 23.4 \\ 
3& C& $\map3$ & DryVR &No  &4& 34.7 \\ 
3& C& $\map4$ & DryVR &No &7& 118.3 \\ 
3& C& $\map1$ & DryVR &Yes  &5& 29.4 \\ 
2& C& $\map1$ & DryVR &No   &5& 21.6 \\ 
2& C& $\map1$ & NeuReach &No   &5& 914.9 \\ 
\hline
\end{tabular}
\end{table}

\subsection{Incremental Verification}
\label{sec:exp:inc}
The incremental verification algorithm \alginc  is evaluated by repeatedly verifying (and simulating) scenarios with slight modifications. In this example, the scenario contains 3 C-CA agents and 5 C-NPV agents in map $\map2$. 
The simulation and verification results for the scenario are shown in Fig.~\ref{fig:exp:inc} in Appendix \ref{app:examples}. The running time for simulation is  19.96s and that for verification is 470.01s. 
We show the results for  3 experiments for both simulation and verification, which corresponds to the 3 rows in each section of Table~\ref{tab:inc-exp}. In the \textit{repeat} experiments, we perform analysis twice on the same scenarios ($\ha_2 = \ha_1$), which shows the most favorable benefits of incremental verification. In the \textit{change init} and \textit{change ctlr} experiments, we respectively change the initial conditions and the decision logic for one of the C-CA agents between the 2 experiment runs. We perform each of these experiments with incremental verification turned off and on, which corresponds to the \verify and \alginc columns. 

From the \yangge{experimental results} we can see that \alginc
indeed reduces the time needed for simulation and verification. This time reduction is also correlated with the similarity between the scenarios, with the \textit{repeat} experiment being able to achieve an almost 10x speedup, while changing the initial conditions leads to almost no benefit. 

\begin{table}[!h]\centering
\caption{\small \yangge{Experimental results} for the Incremental Verification algorithm. The table shows the number of mode transition (\#Tr), run-times in seconds (run time), the memory usage of \ourtool in megabytes (memory), the size of $\guardcache$ and $\flowcache$ in megabytes (cache size), and the hit rate of the caches (hit rate).} 
\label{tab:inc-exp}
\small 
\begin{tabular}{|c|l|c|c|c|c|c|c|c|}
    \hline
    \multicolumn{3}{|c|}{} & \multicolumn{2}{c|}{\verify} & \multicolumn{4}{c|}{\alginc} \\ \hline
    \multicolumn{2}{|c|}{} & \#Tr &  run time & memory & Run time & memory & cache size & hit rate\\ \hline
    \multirow{3}{5.7em}{\textbf{Simulation}}
    & repeat & 45 & 12.18 & 431 & 1.05 & 435 & 3.83 & 83.33\% \\
    & change init & 24 & 10.17 & 431 & 9.3 & 436 & 4.07 & 75.91\% \\
    & change ctlr & 45 & 11.45 & 429 & 5.98 & 438 & 4.38 & 78.19\% \\
    \hline
    \multirow{3}{5.7em}{\textbf{Verification}}
    & repeat & 105 & 450.89 & 498 & 55.34 & 482 & 3.23 & 76.79\% \\
    & change init & 49 & 365.91 & 485 & 349.68 & 498 & 3.7 & 73.21\% \\
    & change ctlr & 93 & 421.65 & 498 & 230.12 & 490 & 4.0 & 73.44\% \\
    \hline
\end{tabular}
\end{table}

\subsection{White-box uncertain continuous dynamics}
\label{sec:uncertainty}
\ourtool provides an implementation of $\postcont$ using algorithms from~\cite{abate_2020,9304461,9304391} which is applicable to agents with white-box  uncertain continuous dynamics. Given a white-box dynamical system defined by a differential equation 
$\dot{x} = f(x,w)$,  where the state $x\in\mathcal{X}\subseteq \reals^n$ and the bounded disturbance input $w \in [w_l,w_u]$. The idea of the approach in~\cite{9304391} is to find a decomposition function $d()$
which can then be used to define an augmented system: 
\[
    \begin{bmatrix}
    \dot{x}_l\\\dot{x}_u
    \end{bmatrix} = \begin{bmatrix}
    d(x_l,w_l,x_u,w_u) \\ d(x_u,w_u,x_l,w_l)
    \end{bmatrix},
\]
such that the reachable set of the original noisy system at time $T$, from any initial set $\mathcal{X}_0 = [\underline{x},\overline{x}]\subseteq \mathcal{X}$  is contained in $[{x_l}(T),{x_u}(T)]$. The latter can be computed by simulating the decomposed system  from the singleton initial state $\langle x_l = \underline{x}, x_u = \overline{x} \rangle$.
%
%
\ourtool currently requires users to provide the decomposition function $d()$ to handle systems with uncertainty in dynamics. Consider an example system and the corresponding manually derived decomposition function: 
\begin{minipage}{0.39\textwidth}
\[
\begin{split}
    & \dot{x}_1 = x_1(1.1+w_1-x_1-0.1x_2)\\
        &\dot{x}_2 = x_2(4+w_2-3x_1-x_2)
\end{split}
\]
\end{minipage}
\begin{minipage}{0.59\textwidth}
\[
d(x,\hat{x},w,\hat{w}) = \begin{bmatrix}
        x_1(1.1+w_1-x_1-0.1\hat{x}_2) \\
        x_2(4+w_2-3\hat{x}_1-x_2)
    \end{bmatrix}
\]
\end{minipage}

With $x_1, x_2\in [0.3, 2]$ and $w_1, w_2\in[-0.1, 0.1]$. \ourtool can automatically construct the augmented system and compute the reachable states for that system (Sample results are shown in Fig.~\ref{fig:exp:uncertain_dynamics} of Appendix~\ref{app:examples}).


\section{Conclusions and future directions}
\label{sec:limits}
In this paper, we presented the new open source \ourtool library for broadening applications of hybrid system verification technologies to scenarios involving multiple interacting decision-making agents. 
\ourtool allows users to  create  agents with decision logics expressed in Python, scenarios with different types of agents, and it  provides functions for performing systematic simulation and verification through reachability analysis. \ourtool maps can be imported from a standard open format, and they allow  agents to be ported across different compatible maps, offering  flexibility in creating scenarios.
The  agent decision logics allow non-deterministic decision making and the reachability functions can propagate the uncertainty in the initial condition through all decision branches. The safety requirements written using $\texttt{assert}$ statements enable various safety conditions. The incremental verification algorithm in \ourtool enables the user to rapidly verify and improve their decision logic. We illustrate useful capabilities and use cases of \ourtool through various examples.

There are several exciting future directions for and around \ourtool. \ourtool currently assumes all agents interact with each other only through the sensor in the scenario and all agents share the same sensor. This restriction could be relaxed  to have different types of asymmetric sensors. 
In incremental verification (and simulation), our \alginc merely opens the door for a \yangge{invention} that  exploit small changes in maps and continuous dynamics.
Functions for constructing and systematically sampling scenarios could be developed. Functions for post-computation for white-box models by building connections with existing tools~\cite{Althoff2015a,flow,FanQM0D16} would be a natural next step. Those approaches could obviously  utilize the symmetry property of agent dynamics as in~\cite{sibai2021scenechecker,sibai-atva-2019}, but beyond that, new types of symmetry reductions should be possibile by exploiting the map geometry.   


%
%
%
%




\bibliographystyle{splncs04}
\bibliography{egbib,sayan1}

\begin{thebibliography}{10}
\providecommand{\url}[1]{\texttt{#1}}
\providecommand{\urlprefix}{URL }
\providecommand{\doi}[1]{https://doi.org/#1}

\bibitem{abate_2020}
Abate, M.: Mixed Monotonicity for Efficient Reachability with Applications to
  Robust Safe Autonomy. Ph.D. thesis, Georgia Institute of Technology (2020)

\bibitem{9304461}
Abate, M., Coogan, S.: Computing robustly forward invariant sets for
  mixed-monotone systems. In: 2020 59th IEEE Conference on Decision and Control
  (CDC). pp. 4553--4559 (2020)

\bibitem{Althoff2015a}
Althoff, M.: An introduction to {CORA} 2015. In: Proc. of the Workshop on
  Applied Verification for Continuous and Hybrid Systems (2015)

\bibitem{alur95algorithmic}
Alur, R., Courcoubetis, C., Halbwachs, N., Henzinger, T.A., Ho, P.H., Nicollin,
  X., Olivero, A., Sifakis, J., Yovine, S.: The algorithmic analysis of hybrid
  systems. Theoretical Computer Science  \textbf{138}(1),  3--34 (1995)

\bibitem{noauthor_asam_nodate}
{Association for Standardization of Automation and Measuring Systems (ASAM)}:
  Open dynamic road information for vehicle environment (Aug 2021),
  \url{https://www.asam.net/standards/detail/opendrive/}

\bibitem{bak2017hylaa}
Bak, S., Duggirala, P.S.: Hylaa: A tool for computing simulation-equivalent
  reachability for linear systems. In: Proceedings of the 20th International
  Conference on Hybrid Systems: Computation and Control. pp. 173--178. ACM
  (2017)

\bibitem{10.1145/3302504.3311792}
Bak, S., Tran, H.D., Johnson, T.T.: Numerical verification of affine systems
  with up to a billion dimensions. In: Proceedings of the 22nd ACM
  International Conference on Hybrid Systems: Computation and Control. p.
  23–32. HSCC '19, Association for Computing Machinery, New York, NY, USA
  (2019)

\bibitem{brittain2022aamgym}
Brittain, M., Alvarez, L.E., Breeden, K., Jessen, I.: {AAM-Gym}: Artificial
  intelligence testbed for advanced air mobility (2022)

\bibitem{flow}
Chen, X., {\'A}brah{\'a}m, E., Sankaranarayanan, S.: Flow*: An analyzer for
  non-linear hybrid systems. In: Computer Aided Verification ({CAV}). pp.
  258--263. Springer (2013)

\bibitem{10.1007/978-3-031-06773-0_6}
Chen, X., Sankaranarayanan, S.: Reachability analysis for cyber-physical
  systems: Are we there yet? In: Deshmukh, J.V., Havelund, K., Perez, I. (eds.)
  NASA Formal Methods. pp. 109--130. Springer, Cham (2022)

\bibitem{9304391}
Coogan, S.: Mixed monotonicity for reachability and safety in dynamical
  systems. In: 2020 59th IEEE Conference on Decision and Control (CDC). pp.
  5074--5085 (2020)

\bibitem{dryvr}
Fan, C., Qi, B., Mitra, S., Viswanathan, M.: Dryvr: Data-driven verification
  and compositional reasoning for automotive systems. In: Majumdar, R.,
  Kun{\v{c}}ak, V. (eds.) Computer Aided Verification ({CAV}). pp. 441--461.
  Springer, Cham (2017)

\bibitem{FanQM0D16}
Fan, C., Qi, B., Mitra, S., Viswanathan, M., Duggirala, P.S.: Automatic
  reachability analysis for nonlinear hybrid models with {C2E2}. In: Computer
  Aided Verification ({CAV}). pp. 531--538 (2016)

\bibitem{utm_conops}
{Federal Aviation Administration}: Unmanned {Aircraft} {System} {Traffic}
  {Management} ({UTM}) {Concept} of {Operations} {Version} 2.0 (Mar 2020)

\bibitem{10.1007/978-3-030-90870-6_20}
Foster, S., Huerta~y Munive, J.J., Gleirscher, M., Struth, G.: Hybrid systems
  verification with {Isabelle/HOL}: Simpler syntax, better models, faster
  proofs. In: Huisman, M., P{\u{a}}s{\u{a}}reanu, C., Zhan, N. (eds.) Formal
  Methods. pp. 367--386. Springer, Cham (2021)

\bibitem{spaceEx}
Frehse, G., Guernic, C.L., Donz{\'e}, A., Cotton, S., Ray, R., Lebeltel, O.,
  Ripado, R., Girard, A., Dang, T., Maler, O.: {SpaceEx}: Scalable verification
  of hybrid systems. In: Computer Aided Verification ({CAV}). pp. 379--395
  (2011)

\bibitem{scenic}
Fremont, D.J., Dreossi, T., Ghosh, S., Yue, X., Sangiovanni-Vincentelli, A.L.,
  Seshia, S.A.: Scenic: A language for scenario specification and scene
  generation. In: Proceedings of the 40th ACM SIGPLAN Conference on Programming
  Language Design and Implementation. p. 63–78. PLDI 2019, Association for
  Computing Machinery, New York, NY, USA (2019)

\bibitem{10.1007/978-3-319-21401-6_36}
Fulton, N., Mitsch, S., Quesel, J.D., V{\"o}lp, M., Platzer, A.: {KeYmaera X}:
  An axiomatic tactical theorem prover for hybrid systems. In: Felty, A.P.,
  Middeldorp, A. (eds.) Automated Deduction - CADE-25. pp. 527--538. Springer,
  Cham (2015)

\bibitem{netcontrol:Mobihoc04}
Goldenberg, D.K., Lin, J., Morse, A.S.: Towards mobility as a network control
  primitive. In: MobiHoc '04: Proceedings of the 5th ACM international
  symposium on Mobile ad hoc networking and computing. pp. 163--174. ACM Press
  (2004)

\bibitem{HENZINGER199894}
Henzinger, T.A., Kopke, P.W., Puri, A., Varaiya, P.: What's decidable about
  hybrid automata? Journal of Computer and System Sciences  \textbf{57}(1),
  94--124 (1998)

\bibitem{4282788}
Hoffmann, G.M., Tomlin, C.J., Montemerlo, M., Thrun, S.: Autonomous automobile
  trajectory tracking for off-road driving: Controller design, experimental
  validation and racing. In: 2007 American Control Conference. pp. 2296--2301
  (2007)

\bibitem{ivanov2019verisig}
Ivanov, R., Weimer, J., Alur, R., Pappas, G.J., Lee, I.: Verisig: verifying
  safety properties of hybrid systems with neural network controllers. In:
  Proceedings of the 22nd ACM International Conference on Hybrid Systems:
  Computation and Control. pp. 169--178 (2019)

\bibitem{TIOAmon}
Kaynar, D.K., Lynch, N., Segala, R., Vaandrager, F.: The Theory of Timed {I/O}
  Automata. Synthesis Lectures on Computer Science, Morgan Claypool (November
  2005), also available as Technical Report MIT-LCS-TR-917, MIT

\bibitem{LKLM:formats05}
Lim, H., Kaynar, D., Lynch, N., Mitra, S.: Translating timed {I/O} automata
  specifications for theorem proving in {PVS}. In: Proceedings of Formal
  Modelling and Analysis of Timed Systems ({FORMATS'05}). No.~3829 in LNCS,
  Springer, Uppsala, Sweden (September 2005)

\bibitem{SUMO2018}
Lopez, P.A., Behrisch, M., Bieker-Walz, L., Erdmann, J., Fl{\"o}tter{\"o}d,
  Y.P., Hilbrich, R., L{\"u}cken, L., Rummel, J., Wagner, P., Wie{\ss}ner, E.:
  Microscopic traffic simulation using {SUMO}. In: The 21st IEEE International
  Conference on Intelligent Transportation Systems. IEEE (2018)

\bibitem{manfredi2016introduction}
Manfredi, G., Jestin, Y.: An introduction to {ACAS Xu} and the challenges
  ahead. In: 2016 IEEE/AIAA 35th Digital Avionics Systems Conference (DASC).
  pp.~1--9. IEEE (2016)

\bibitem{GRAICrace}
{Minghao Jiang and Zexiang Liu and Kristina Miller and Dawei Sun and Arnab
  Datta and Yixuan Jia and Sayan Mitra and Necmiye Ozay}: Graic: A simulator
  framework for autonomous racing. \url{https://popgri.github.io/Race/} (2021)

\bibitem{python_grammar}
{Python Software Foundation}: Python full grammar specification (Oct 2022)

\bibitem{10.1007/978-3-319-26287-1_1}
Ray, R., Gurung, A., Das, B., Bartocci, E., Bogomolov, S., Grosu, R.: Xspeed:
  Accelerating reachability analysis on multi-core processors. In: Piterman, N.
  (ed.) Hardware and Software: Verification and Testing. pp. 3--18. Springer,
  Cham (2015)

\bibitem{sibai2021scenechecker}
Sibai, H., Li, Y., Mitra, S.: {SceneChecker}: Boosting scenario verification
  using symmetry abstractions (2021)

\bibitem{sibai-tacas-2020}
Sibai, H., Mokhlesi, N., Fan, C., Mitra, S.: Multi-agent safety verification
  using symmetry transformations. In: Biere, A., Parker, D. (eds.) Tools and
  Algorithms for the Construction and Analysis of Systems. pp. 173--190.
  Springer, Cham (2020)

\bibitem{sibai-atva-2019}
Sibai, H., Mokhlesi, N., Mitra, S.: Using symmetry transformations in
  equivariant dynamical systems for their safety verification. In: Chen, Y.F.,
  Cheng, C.H., Esparza, J. (eds.) Automated Technology for Verification and
  Analysis. pp. 98--114. Springer, Cham (2019)

\bibitem{neuReach-SunM22}
Sun, D., Mitra, S.: Neureach: Learning reachability functions from simulations.
  In: Tools and Algorithms for the Construction and Analysis of Systems - 28th
  International Conference, {TACAS} 2022, Held as Part of the European Joint
  Conferences on Theory and Practice of Software, {ETAPS} 2022, Munich,
  Germany, April 2-7, 2022, Proceedings, Part {I}. pp. 322--337 (2022)

\bibitem{Wu2017FlowAA}
Wu, C., Kreidieh, A., Parvate, K., Vinitsky, E., Bayen, A.M.: Flow:
  Architecture and benchmarking for reinforcement learning in traffic control.
  ArXiv  \textbf{abs/1710.05465} (2017)

\bibitem{OLVECZKY2002359}
Ölveczky, P.C., Meseguer, J.: Specification of real-time and hybrid systems in
  rewriting logic. Theoretical Computer Science  \textbf{285}(2),  359--405
  (2002), rewriting Logic and its Applications

\end{thebibliography}

\appendix
\newpage

\section{Example Maps}
The figure for maps used in the examples in Section \ref{sec:exp} are shown in Fig. \ref{fig:map}. 
\begin{figure}[!htb]
     \begin{subfigure}{0.48\textwidth}
        \centering
        \includegraphics[width=\textwidth]{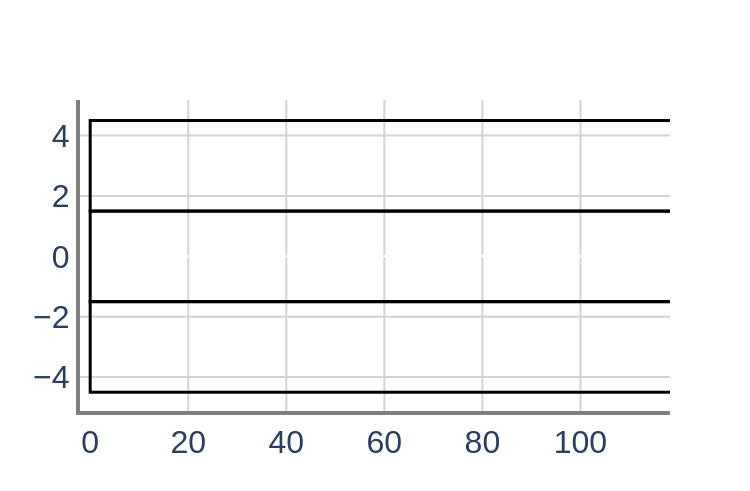}
        \caption{$\map1$}\label{fig:map:m1}
    \end{subfigure}\hfill
    \begin{subfigure}{0.48\textwidth}
        \centering
        \includegraphics[width=\textwidth]{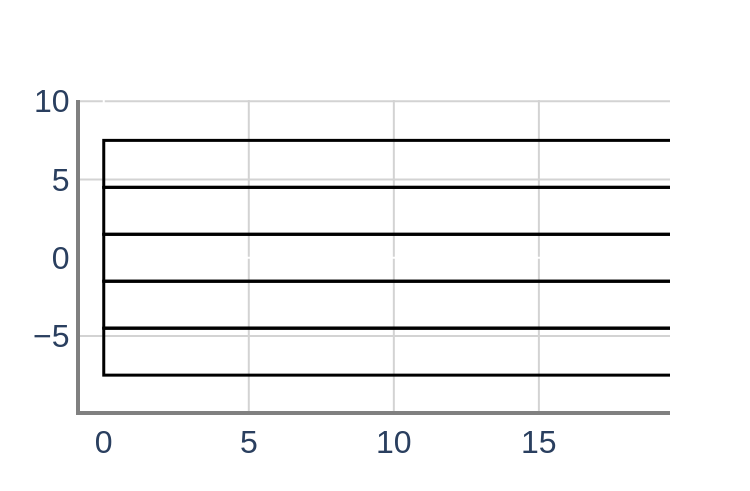}
        \caption{$\map2$}\label{fig:map:m2}
    \end{subfigure}\hfill
    \begin{subfigure}{0.48\textwidth}
        \centering
        \includegraphics[width=\textwidth]{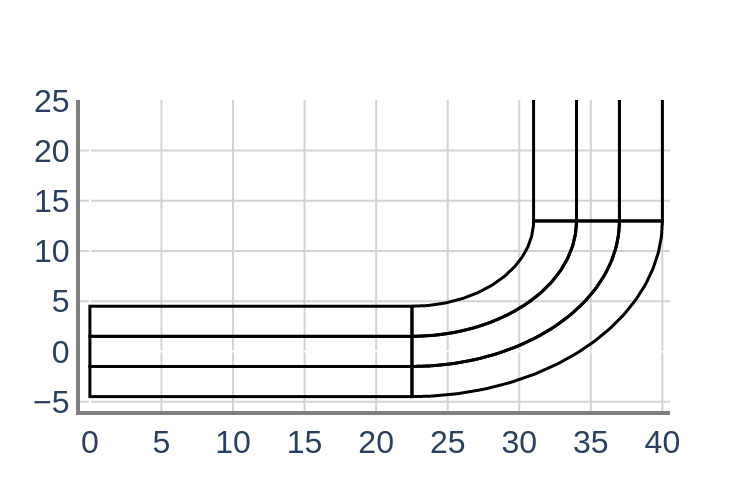}
        \caption{$\map3$}\label{fig:map:m3}
    \end{subfigure}\hfill
    \begin{subfigure}{0.48\textwidth}
        \centering
        \includegraphics[width=\textwidth]{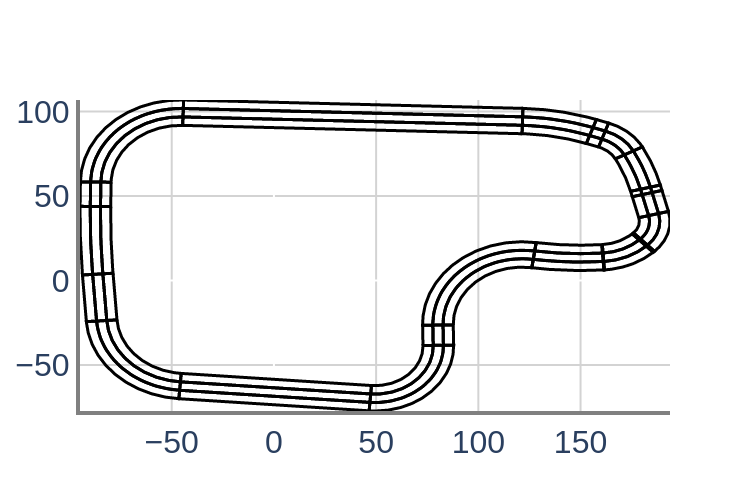}
        \caption{$\map4$}\label{fig:map:m4}
    \end{subfigure}\hfill
    \begin{subfigure}{0.48\textwidth}
        \centering
        \includegraphics[width=\textwidth]{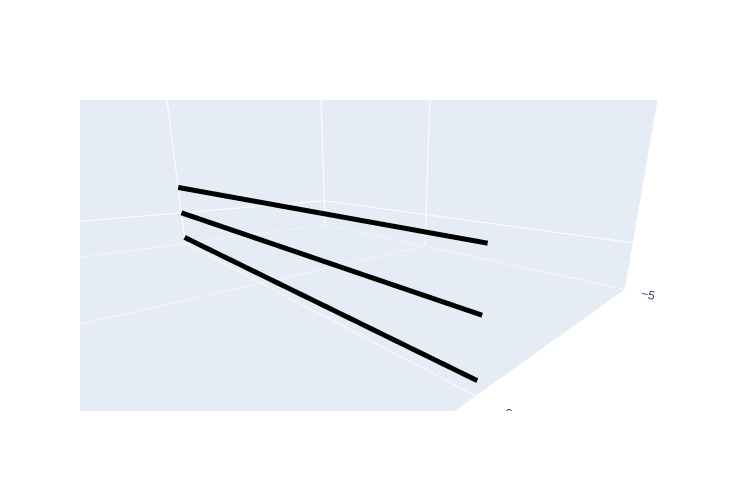}
        \caption{$\map5$}\label{fig:map:m5}
    \end{subfigure}\hfill
    \begin{subfigure}{0.48\textwidth}
        \centering
        \includegraphics[width=\textwidth]{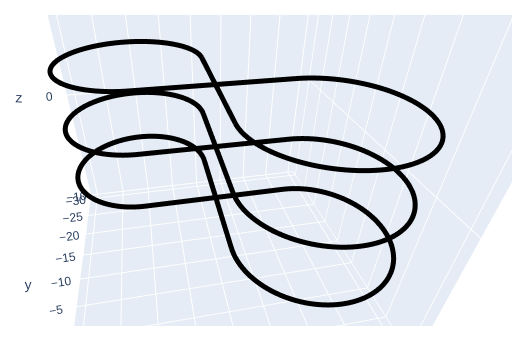}
        \caption{$\map6$}\label{fig:map:m6}
    \end{subfigure}\hfill
\caption{Maps in examples}\label{fig:map}
\end{figure}

\section{Grammar for Decision Logic Code}\label{sec:parser_grammar}

This is the subset of the Python grammar~\cite{python_grammar} that \ourtool support for coding the decision logic of agents. Some details about operator precedence are not given here. Currently, the supported \texttt{OPERATOR} include logic operators (\texttt{and}, \texttt{or}, and \texttt{not}), comparison operators (\texttt{<=}, \texttt{<}, \texttt{==}, \texttt{!=}, \texttt{>}, and \texttt{>=}), and arithmetic operators (\texttt{+}, \texttt{-}, \texttt{*} and \texttt{/}).


\begin{figure}
    \centering
    \scriptsize
\begin{minted}{peg.py:CustomPegLexer -x}
expression:
    | expression OPERATOR expression
    | 'lambda' parameters ':' expression
    | literal
    | dotted_name tuple
literal:
    | ['-'] NUMBER
    | STRING+
    | 'None' | 'True' | 'False' 
    | tuple | group | genexp
expressions: ','.expression*
tuple: '(' expressions ')'
group: '(' expression ')'
genexp: '(' expression for_if_clauses+ ')'
for_if_clause: 'for' IDENT 'in' expression ('if' expression )*
function_def: 'def' IDENT '(' [params] ')' ['->' expression ] ':' block
params: (maybe_typed_name ',')* maybe_typed_name?
maybe_typed_name: IDENT [':' (expression | STRING)]
block:
    | NEWLINE INDENT statements DEDENT
    | simple_stmts
statements: statement+
statement: compound_stmt | simple_stmts
statement_newline: compound_stmt NEWLINE | simple_stmts | NEWLINE | EOF
compound_stmt: function_def | if_stmt | class_def
simple_stmts: ';'.simple_stmt+ [';'] NEWLINE 
simple_stmt: assignment | return_stmt | assert_stmt | if_stmt
assignment: IDENT '=' expression
return_stmt: 'return' [expression]
assert_stmt: 'assert' expression [',' expression ] 
if_stmt: 'if' expression ':' block [else_block] 
else_block: 'else' ':' block
dotted_as_names: ','.dotted_as_name+ 
dotted_as_name: dotted_name ['as' NAME ] 
dotted_name: dotted_name '.' NAME | NAME

\end{minted}
\end{figure}

\section{Additional Examples}
The verification result for a 3-car scenario with a map imported from OpenDRIVE format is shown in Fig. \ref{fig:exp:opendrive}.
\label{app:examples}
\begin{figure}
    \centering
    \includegraphics[width=0.48\textwidth]{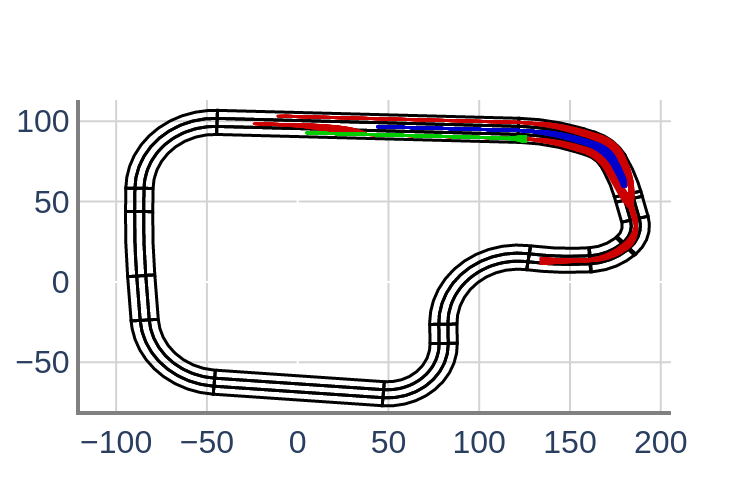} 
    \includegraphics[width=0.48\textwidth]{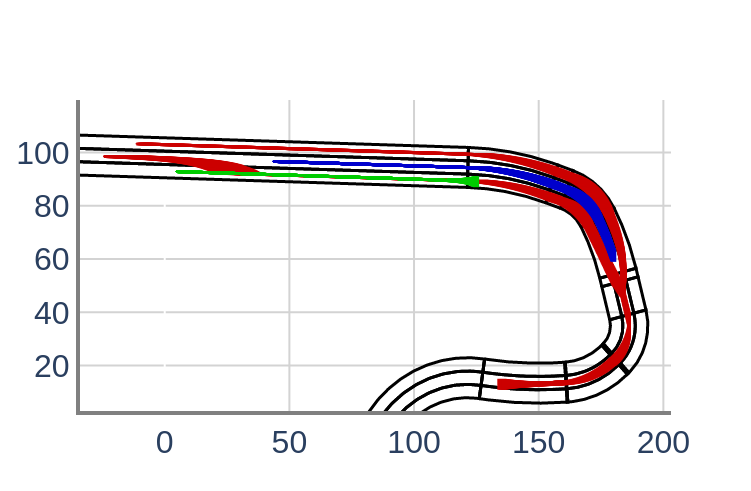}
    \caption{Picture showing verification result for a scenario with map imported from OpenDRIVE format. }
    \label{fig:exp:opendrive}
\end{figure}

The verification result with $\postcont$ computed by NeuReach and DryVR is shown in Fig. \ref{fig:exp:neureach_post}.
\begin{figure}
    \centering
    \includegraphics[width=0.48\textwidth]{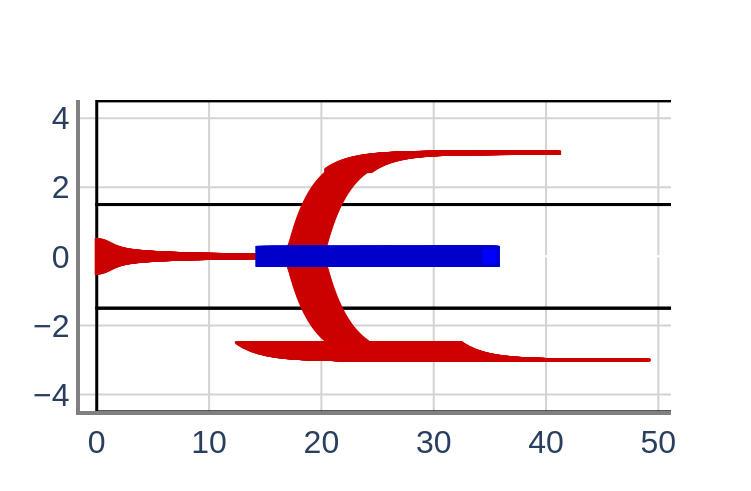} 
    \includegraphics[width=0.48\textwidth]{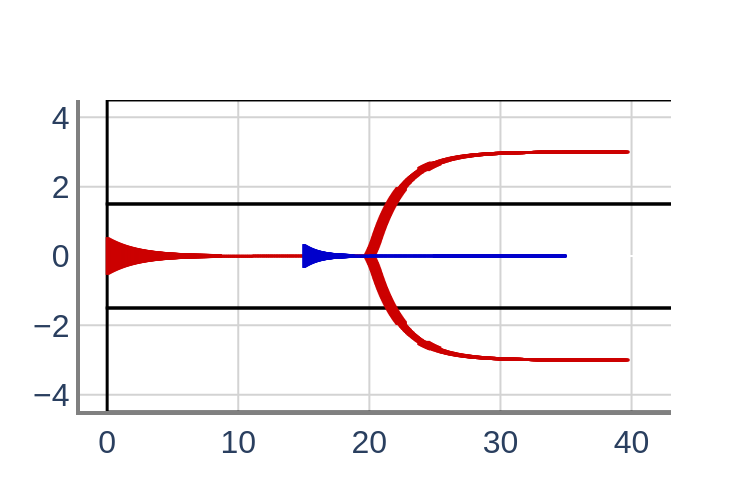}
    \caption{Picture on the left showing scenario with $\postcont$ computed by NueReach. Picture on the right showing same scenario with $\postcont$ computed by DryVR. }
    \label{fig:exp:neureach_post}
\end{figure}

The ploted reachtube for $x_1$ and $x_2$ for system with uncertain dynamics in Section \ref{sec:uncertainty} is shown in Fig. \ref{fig:exp:uncertain_dynamics}.
\begin{figure}
    \centering
    \includegraphics[width=0.48\textwidth]{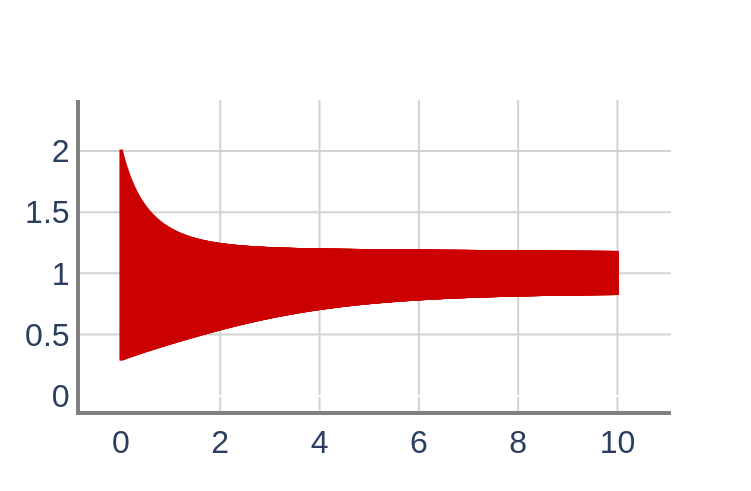}
    \includegraphics[width=0.48\textwidth]{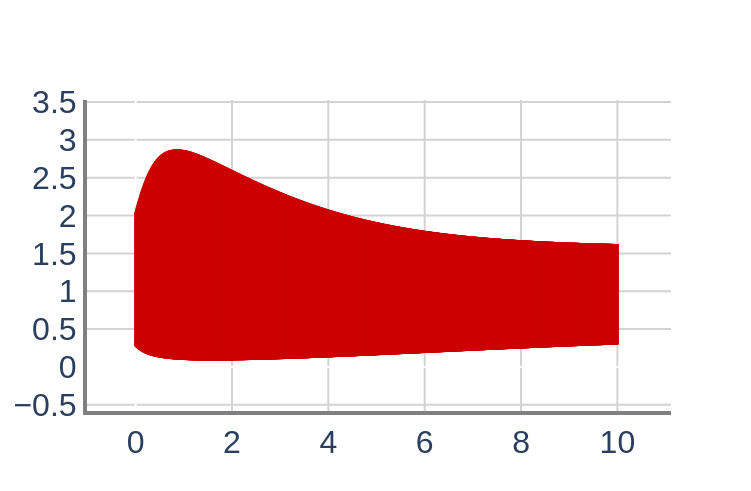}
    \caption{Reachtube plot for $x_1$ (left) and $x_2$ (right) for the example system.}
    \label{fig:exp:uncertain_dynamics}
\end{figure}

The reachtube plot for the 7-car scenario described in row 5 of Table \ref{tab:exps} is shown in Fig. \ref{fig:exp:7-car}. 
\begin{figure}
    \centering
    \includegraphics[width=0.48\textwidth]{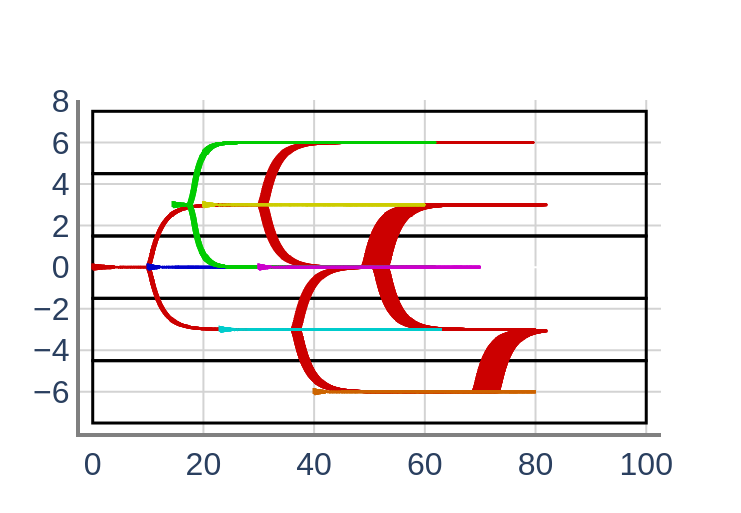}
    \caption{Reachtube plot for 7-car scenario}
    \label{fig:exp:7-car}
\end{figure}

The simulation and verification result for the baseline 8-car system used for experimenting incremental verification in Section \ref{sec:exp:inc} is shown in Fig. \ref{fig:exp:inc}
\begin{figure}
    \centering
    \includegraphics[width=0.48\textwidth]{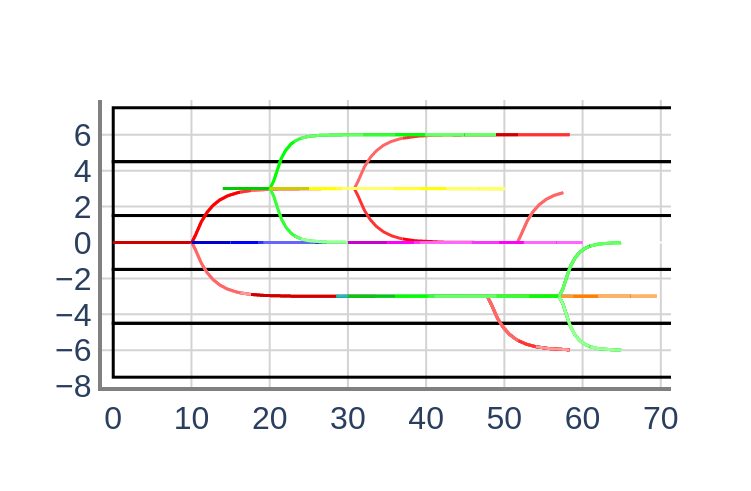}
    \includegraphics[width=0.48\textwidth]{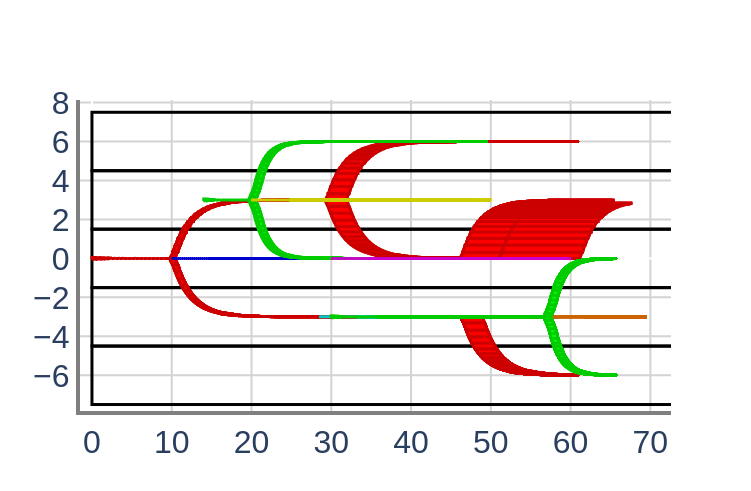}
    \caption{Simulation (left) and reachtube plot (right) for scenario in incremental verification experiment.}
    \label{fig:exp:inc}
\end{figure}

The baseline scenario for the incremental verification experiments is specified in Fig. \ref{fig:inc-scene-base-sim} for simulation and in Fig. \ref{fig:inc-scene-base-veri} for verification. In the \textit{change init} experiment, the initial condition for \texttt{car7} is changed to \texttt{[[50, -3, 0, 0.5], [50, -3, 0, 0.5]]}. In the \textit{change ctlr} experiment, the controller for \texttt{car8} is changed so that it switches tracks if there is another car less than 4.5 meters in front of it instead of 5 meters.
\begin{figure}
    \centering
    \scriptsize
    \caption{The scenario specification for the simulation baseline in the incremental verification experiments}
    \label{fig:inc-scene-base-sim}
    \begin{minted}{python}
    scenario = Scenario()
    controller_file = '...'
    car1 = CarAgent('car1', file_name=controller_file)
    car1.set_initial([[0, 0, 0, 1.0], [0, 0, 0, 1.0]], (TacticalMode.Normal, TrackMode.T1))
    scenario.add_agent(car1)
    car2 = NPCAgent('car2')
    car2.set_initial([[10, 0, 0, 0.5], [10, 0, 0, 0.5]], (TacticalMode.Normal, TrackMode.T1))
    scenario.add_agent(car2)
    car3 = CarAgent('car3', file_name=controller_file)
    car3.set_initial([[14, 3, 0, 0.6], [14, 3, 0, 0.6]], (TacticalMode.Normal, TrackMode.T0))
    scenario.add_agent(car3)
    car4 = NPCAgent('car4')
    car4.set_initial([[20, 3, 0, 0.5], [20, 3, 0, 0.5]], (TacticalMode.Normal, TrackMode.T0))
    scenario.add_agent(car4)
    car5 = NPCAgent('car5')
    car5.set_initial([[30, 0, 0, 0.5], [30, 0, 0, 0.5]], (TacticalMode.Normal, TrackMode.T1))
    scenario.add_agent(car5)
    car6 = NPCAgent('car6')
    car6.set_initial([[28.5, -3, 0, 0.5], [28.5, -3, 0, 0.5]], (TacticalMode.Normal, TrackMode.T2))
    scenario.add_agent(car6)
    car7 = NPCAgent('car7')
    car7.set_initial([[39.5, -3, 0, 0.5], [39.5, -3, 0, 0.5]], (TacticalMode.Normal, TrackMode.T2))
    scenario.add_agent(car7)
    car8 = CarAgent('car8', file_name=controller_file)
    car8.set_initial([[30, -3, 0, 0.6], [30, -3, 0, 0.6]], (TacticalMode.Normal, TrackMode.T2))
    scenario.add_agent(car8)
    \end{minted}
\end{figure}

\begin{figure}
    \centering
    \scriptsize
    \caption{The scenario specification for the verification baseline in the incremental verification experiments}
    \label{fig:inc-scene-base-veri}
    \begin{minted}{python}
    scenario = Scenario()
    controller_file = '...'
    car1 = CarAgent('car1', file_name=controller_file)
    car1.set_initial([[0, -0.05, 0, 1.0], [0, 0.05, 0, 1.0]], (TacticalMode.Normal, TrackMode.T1))
    scenario.add_agent(car1)
    car2 = NPCAgent('car2')
    car2.set_initial([[10, 0, 0, 0.5], [10, 0, 0, 0.5]], (TacticalMode.Normal, TrackMode.T1))
    scenario.add_agent(car2)
    car3 = CarAgent('car3', file_name=controller_file)
    car3.set_initial([[14, 2.95, 0, 0.6], [14, 3.05, 0, 0.6]], (TacticalMode.Normal, TrackMode.T0))
    scenario.add_agent(car3)
    car4 = NPCAgent('car4')
    car4.set_initial([[20, 3, 0, 0.5], [20, 3, 0, 0.5]], (TacticalMode.Normal, TrackMode.T0))
    scenario.add_agent(car4)
    car5 = NPCAgent('car5')
    car5.set_initial([[30, 0, 0, 0.5], [30, 0, 0, 0.5]], (TacticalMode.Normal, TrackMode.T1))
    scenario.add_agent(car5)
    car6 = NPCAgent('car6')
    car6.set_initial([[28.5, -3, 0, 0.5], [28.5, -3, 0, 0.5]], (TacticalMode.Normal, TrackMode.T2))
    scenario.add_agent(car6)
    car7 = NPCAgent('car7')
    car7.set_initial([[39.5, -3, 0, 0.5], [39.5, -3, 0, 0.5]], (TacticalMode.Normal, TrackMode.T2))
    scenario.add_agent(car7)
    car8 = CarAgent('car8', file_name=controller_file)
    car8.set_initial([[30, -3.05, 0, 0.6], [30, -2.95, 0, 0.6]], (TacticalMode.Normal, TrackMode.T2))
    scenario.add_agent(car8)
    \end{minted}
\end{figure}

The simulation and verification result for the 8-car system after changing initial condition is shown in Fig. \ref{fig:exp:inc_init}. The agent with changed controller is marked with the red box in the plot. 
\begin{figure}
    \centering
    \includegraphics[width=0.48\textwidth]{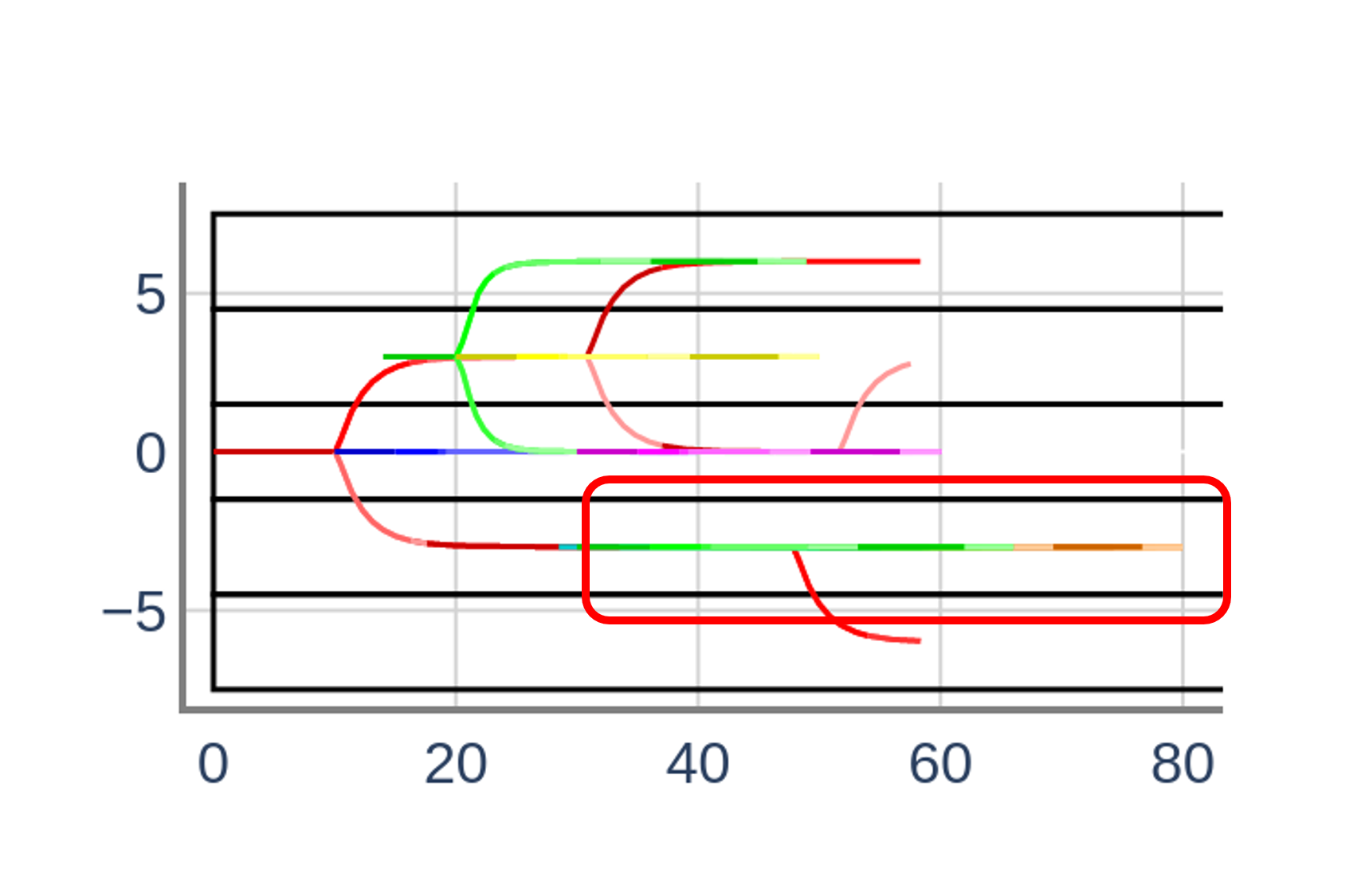}
    \includegraphics[width=0.48\textwidth]{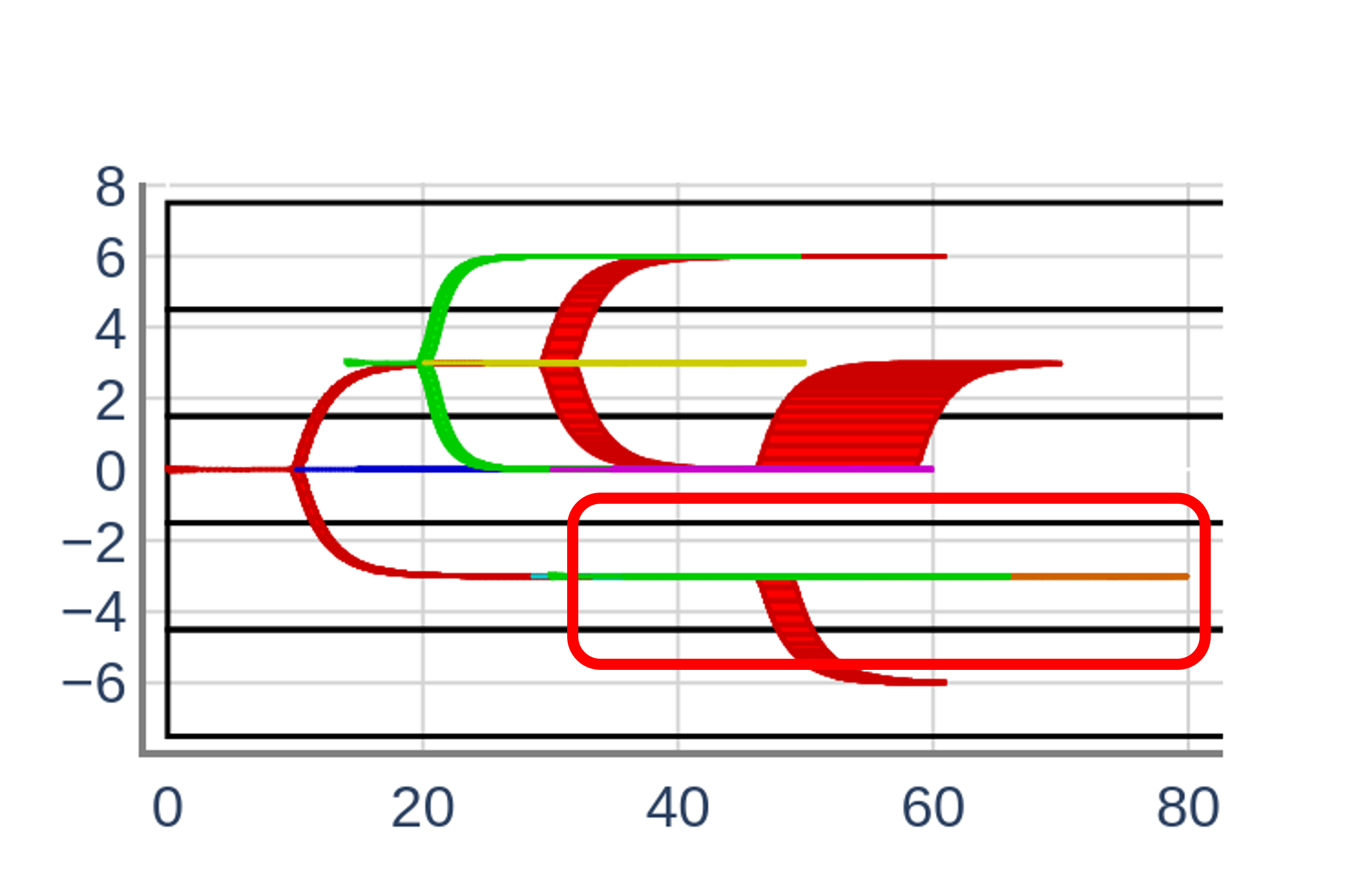}
    \caption{Simulation (left) and reachtube plot (right) for scenario after changing initial condition in incremental verification experiment.}
    \label{fig:exp:inc_init}
\end{figure}

The simulation and verification result for the 8-car system after changing controller is shown in Fig. \ref{fig:exp:inc_ctrl}. The agent with changed controller is marked with the red box in the plot.
\begin{figure}
    \centering
    \includegraphics[width=0.48\textwidth]{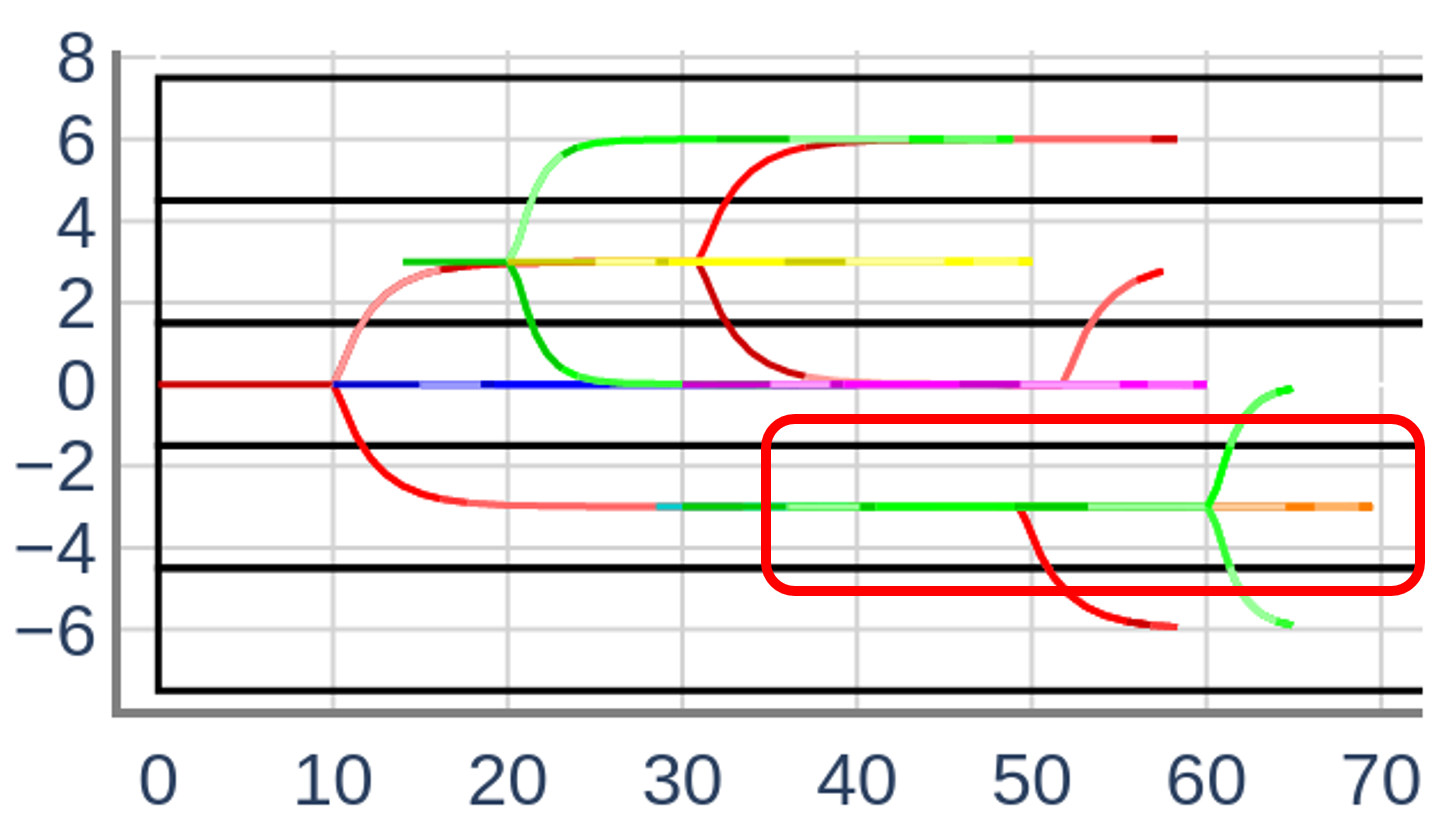}
    \includegraphics[width=0.48\textwidth]{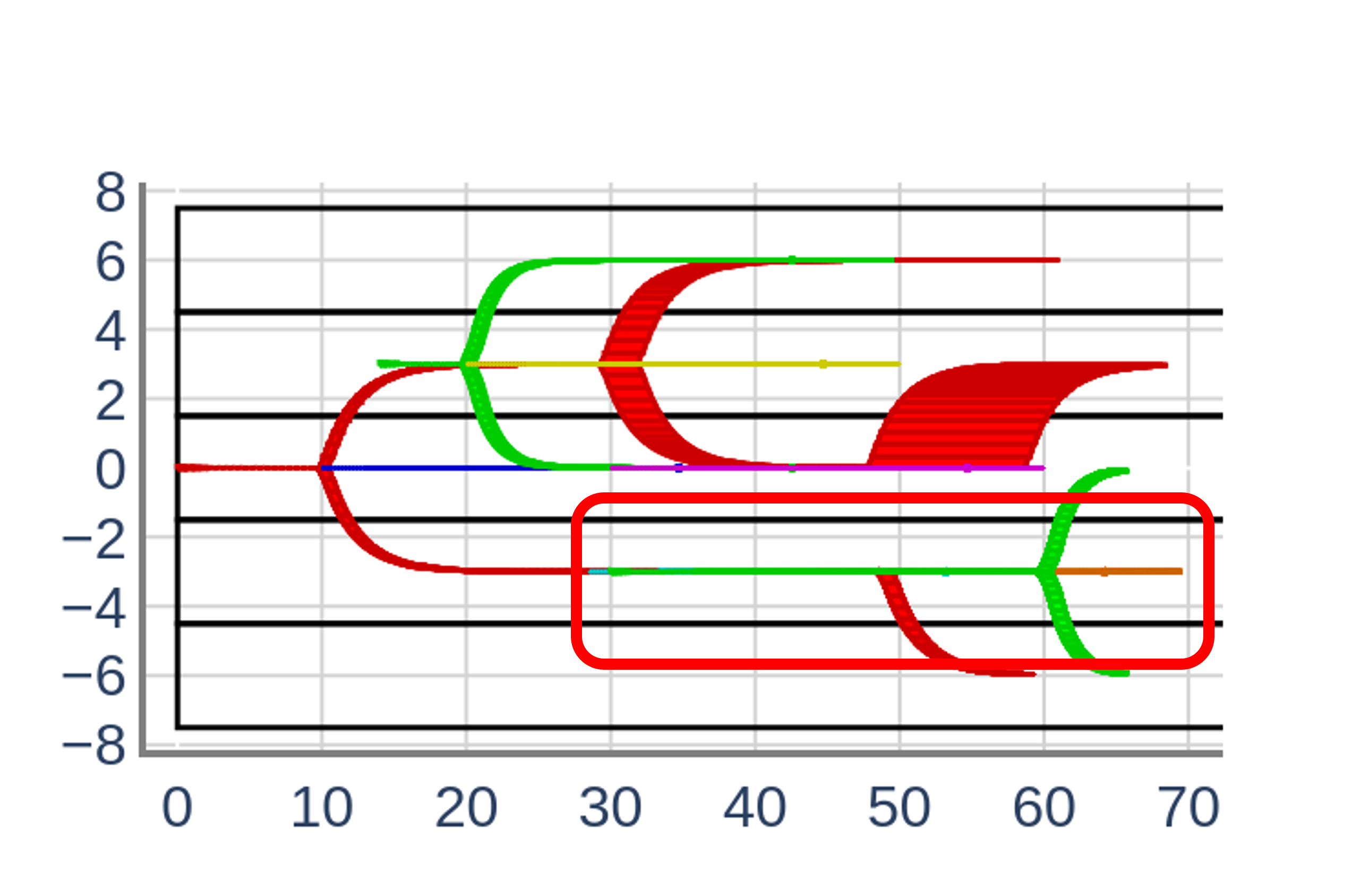}
    \caption{Simulation (left) and reachtube plot (right) for scenario after changing controller in incremental verification experiment.}
    \label{fig:exp:inc_ctrl}
\end{figure}

\end{document}